\documentclass[11pt]{article}
\usepackage{amssymb,amsmath,amsthm,amsfonts,bbm}
\usepackage{mathtools}
\usepackage{enumerate}
\usepackage{tikz,listings}
\usetikzlibrary {patterns,patterns.meta}
\usepackage{pgfplots}
\usepackage{tikz}
\usepackage{pgfplots}
\pgfplotsset{compat=1.18}
\usetikzlibrary{arrows.meta, positioning, calc}
\usepackage{subfigure}
\usepackage[capposition=top]{floatrow}
\usepackage[margin=1in]{geometry}
\usepackage{natbib}
\usepackage{verbatim}  % for comment
\usepackage{graphicx,booktabs}
\usepackage[outline]{contour}
\usepackage{datetime}
\usepackage[]{hyperref}                                  % Allow hyperlinks (internal and external)
\hypersetup{                                             % Custom hyperlink settings
    pdffitwindow=false,                                  % Window fit to page when opened
    pdfstartview={XYZ null null 1.00},                   % Fits the zoom of the page to 100%
    pdfnewwindow=true,                                   % Links in new window
    colorlinks=true,                                     % false: boxed links; true: colored links
    linkcolor=blue,                                      % Color of internal links (black is necessary for printing quality)
    citecolor=blue,                                      % Color of links to bibliography
    urlcolor=blue                                        % Color of external links
}
\usepackage{graphicx}
\usepackage{mathrsfs}
\usepackage{dsfont}
\usepackage{graphicx}
\usepackage{units}
\usepackage{amsmath, amsthm, amssymb, amsfonts, enumerate}
\usepackage{natbib}
\usepackage{color}
\usepackage{caption}
\usepackage{subcaption}
\usepackage{setspace}
\usepackage{mdframed}
\usepackage{hyperref}
\usepackage{float}
\usepackage{multirow} % Add this to the preamble for the multirow feature
\usepackage[normalem]{ulem}
\usepackage{placeins}
\usepackage{tikz}
\usetikzlibrary{calc}
\usepackage{cleveref}
\usepackage{titlesec}
\usetikzlibrary{calc}

\definecolor{Red}{rgb}{1, 0.0, 0}
\definecolor{blue}{rgb}{0.0, 0.0, 1}
\definecolor{green2}{rgb}{0.0, 0.52, 0.24}
\definecolor{cadmiumgreen}{rgb}{0.0, 0.42, 0.24}
\definecolor{camouflagegreen}{rgb}{0.47, 0.53, 0.42}
\definecolor{darkolivegreen}{rgb}{0.33, 0.42, 0.18}
\definecolor{darkpastelgreen}{rgb}{0.01, 0.75, 0.24}
\definecolor{darkspringgreen}{rgb}{0.09, 0.45, 0.27}
\definecolor{darkspringgreen}{rgb}{0.09, 0.45, 0.27}

\usepackage{comment} %commenting bubbles package.
\usepackage[textwidth=30mm,disable]{todonotes}

\newcommand{\E}{\mathbb{E}}

\newcommand{\Red}[1]{{{\textcolor{red}{#1}}}}

\newcommand{\heading}[1]{\vskip .3cm \noindent \textbf {#1}. \hskip .2cm }

\def\p{{\bf p}}
\def\p{{\bf p}}

\def\P{\mathbb{P}}

\newcommand{\Prob}{\mathbb{P}}

\long\def\ignore#1{}
\def\a{\alpha}

\def\?{\Red{?????????}}

\def\1{\emph{\textbf{1}}}
\def\eps{\varepsilon}

\titlespacing*{\section}{0pt}{*0.7}{*0.4}
\titlespacing*{\subsection}{0pt}{*0.4}{*0.2}
\usepackage{pgfplots}
\pgfplotsset{compat=1.18}

\definecolor{blue}{rgb}{0.0, 0.0, 1}
\definecolor{green2}{rgb}{0.0, 0.82, 0.24}
\definecolor{cadmiumgreen}{rgb}{0.0, 0.42, 0.24}
\definecolor{camouflagegreen}{rgb}{0.47, 0.53, 0.42}
\definecolor{darkolivegreen}{rgb}{0.33, 0.42, 0.18}
\definecolor{darkpastelgreen}{rgb}{0.01, 0.75, 0.24}
\definecolor{darkspringgreen}{rgb}{0.09, 0.45, 0.27}
\definecolor{darkspringgreen}{rgb}{0.09, 0.45, 0.27}

\usepackage{comment} %commenting bubbles package.
\usepackage[textwidth=30mm]{todonotes}

\titlespacing*{\section}{0pt}{*0.7}{*0.4}
\titlespacing*{\subsection}{0pt}{*0.4}{*0.2}

\colorlet{dred}{red!50!black}
\colorlet{dgreen}{green!80!black}
\newdateformat{mydate}{\monthname[\THEMONTH] \THEDAY, \THEYEAR}
\allowdisplaybreaks

\theoremstyle{plain}% default
\newtheorem{theorem}{Theorem}
\newtheorem{lemma}{Lemma}
\newtheorem{proposition}{Proposition}
\newtheorem{corollary}{Corollary}

\theoremstyle{definition}
\newtheorem{definition}{Definition}
\newtheorem{example}{Example}
\newtheorem{assumption}{Assumption}

\theoremstyle{remark}
\newtheorem{remark}{Remark}

\newtheoremstyle%
{bluethm}%
{}{}%
{\color{blue}}%\itshape
{}%
{\color{blue}\bfseries}
{\color{blue}.}%
{ }{}

\theoremstyle{bluethm}

\newcommand{\ind}{\mathbbm{1}}

\newcommand{\EL}[1]{\fcolorbox{red}{yellow}{\textcolor{red}{\textbf{Ehud:} #1}}}

\begin{document}

\title{Bayesian Sequential Search with Censored Observations\thanks{
We thank Peter Achim, Eren Arbatli, Guilherme Carmona, Remzi Kaygusuz, Joosung Lee, Gilat Levy, Leslie Reinhorn, Tao Wang for helpful comments and discussions, and the participants of the 6th DETC conference, the 2026 Conference on Mechanism and Institution Design, the AMES 2026, as well as the seminar participants at Durham and CUEB for valuable comments.}}

\author{Ehud Lehrer\footnote{Durham University Business School. Email: \href{ehud.m.lehrer@durham.ac.uk}{ehud.m.lehrer@durham.ac.uk}}
		\and
	Daniel Z. Li\footnote{Durham University Business School. Tel: 01913346335. Email: \href{daniel.li@durham.ac.uk}{daniel.li@durham.ac.uk}.} }
 \maketitle
%\centerline{Preliminary draft}

\begin{abstract} 
\noindent
This paper studies how information censoring enables a myopic cutoff rule in Bayesian sequential search. Under full information, Bayesian learning generally destroys the monotonicity of continuation values, preventing simple cutoff rules. We show that one-sided censoring restores monotonicity by limiting posterior fluctuations, thereby making a myopic cutoff rule optimal. By decomposing the intertemporal change in the marginal value of search into a fallback-value effect and a learning effect, we derive necessary and sufficient conditions for monotonicity under lower censoring and characterize the optimal cutoff rule. In contrast, under full revelation, monotonicity requires highly restrictive conditions. We further show that expected monotonicity (i.e., the supermartingale property) is characterized by the same conditions under both lower censoring and full revelation, owing to Bayes plausibility and the affine structure of the problem. Thus, censoring restores monotonicity not by altering expected learning, but by reducing posterior volatility. Finally, we apply our framework to job search, consumer price search, and product experimentation.

\vspace{0.2cm}
\noindent \textbf{Keywords:} sequential search; Bayesian learning; pathwise monotonicity; expected monotonicity; fallback value effect; learning effect; censored observation. 

\vspace{0.2cm}
\noindent \textbf{JEL classification:} C61; D83. 

\end{abstract}

\newpage

\pagebreak  \setcounter{page}{1}
%\linespread{1.5}

\normalsize
\section{Introduction}
Search problems with learning arise in many economic settings in which a decision maker (DM) sequentially explores uncertain opportunities while updating her beliefs based on observed outcomes. Classical search theory shows that when the payoff distribution is known, the optimal stopping rule is myopic and characterized by reservation cutoffs \citep{mccall1970economics,weitzman1979optimal}. Under free recall, the DM's fallback value, the best payoff observed so far, is weakly increasing over time. As a result, the expected marginal benefit of an additional search weakly decreases, making the search problem \textit{monotone}. When the payoff distribution is unknown, however, each observation can affect not only the fallback value but also the DM's beliefs about the distribution of future draws. This learning effect may destroy monotonicity, rendering myopic and simple reservation-cutoff policies no longer optimal.

This paper studies Bayesian search with censored observations, whereby the DM does not observe the outcomes of all previous searches. Such censoring arises naturally in many economic environments where information is revealed selectively. For example, rejected job applicants typically learn only that they failed to meet an employer's acceptance threshold; consumers investigate a product's detailed attributes only after it satisfies an initial price screen; and investors obtain precise information about a project only if it exceeds the quality of the best prototype considered so far. As a result, observations that fail to pass the relevant benchmark convey only partial information, making learning intrinsically censored.

Our main result establishes that censored learning can restore monotonicity in Bayesian search problems under mild distributional conditions, thereby making the problem tractable. To identify the underlying mechanism, we decompose the intertemporal change in the continuation value of one additional search into two components: a fallback-value effect and a learning effect. Monotonicity follows whenever their combined effect is non-positive for every possible history. The fallback-value effect is always non-positive because the fallback option, the best payoff observed so far, weakly increases over time, reducing the expected marginal benefit of further search. The learning effect, however, may be either positive or negative, depending on whether new information shifts beliefs toward more or less favorable payoff distributions.

In general, the learning effect can dominate the fallback-value effect and destroy monotonicity. One-sided censoring, for instance, revealing only outcomes that exceed the current maximum, restores monotonicity by reducing the informativeness of observations and limiting posterior dispersion. Because all non-improving outcomes are pooled into a single signal, the censored posterior is a martingale projection of the full-information posterior: it is a compressed conditional expectation of the belief update under full revelation. This pooling of the lower tail restricts the volatility of beliefs and, by Jensen's inequality, bounds the extent of favorable belief revisions.  As a result, the negative fallback-value effect consistently dominates the learning effect along every feasible search history, ensuring that simple, intuitive reservation-cutoff policies remain optimal.

Our baseline model considers two candidate discrete distributions, $L_1, L_2$, for the data-generating distribution. 
The DM holds a posterior belief $\alpha$ and updates her beliefs according to Bayes' rule. 
Under lower censoring, outcomes that do not exceed the current fallback value $y$ are pooled into the censored event $\{X\le y\}$, while breakthrough outcomes satisfying $X>y$ are fully revealed. 
The updated fallback value is therefore $\max\{y,X\}$, and the one-step marginal search value is
\[
    U(y,\alpha)=\mathbb{E}_{\alpha}[(X-y)_+]
    =\alpha U_1(y)+(1-\alpha)U_2(y),
\]
where $U_i(y)$ denotes the one-step marginal search value when the payoff distribution is distribution $i$. 
A key observation is that $U(y,\alpha)$ is affine in the posterior belief $\alpha$.

The search problem is monotone if $U(y,\alpha)$ weakly decreases after every feasible search history.  Assuming a single-crossing order between the two candidate distributions, we derive a necessary and sufficient condition for monotonicity (Theorem~\ref{thm: monotone conditions_lower}).  We further provide a sufficient condition that depends only on the primitive distributions and is therefore straightforward to verify in practice.  Under these conditions, the optimal stopping rule is characterized either by a cutoff value $y^*(\alpha)$ or equivalently, by a cutoff belief $\underline\alpha(y)$. 
Moreover, the number of searches required by the optimal policy is uniformly bounded across all feasible histories (Proposition~\ref{prop:finite-pathwise-bound}).  The baseline model can be readily extended to continuous distributions (Theorem \ref{thm: monotone conditions_lower_continuous}).

We illustrate the applicability of our results through two economic applications. 
The first is job search with recruitment screening, where a lower-censoring information structure naturally arises from the recruiter's perspective. 
The second is consumer price search, in which consumers search for lower prices. 
This setting provides the upper-censoring counterpart to our baseline model of lower censoring: the fallback value is the lowest price observed so far, denoted by $p$, and price quotes weakly above $p$ are censored. 
The same underlying logic applies. 
Maintaining the single-crossing assumption between the two distributions, we derive conditions under which the price-search problem is monotone (Proposition~\ref{prop: mono-condition-upper}).

We extend our analysis in several directions. 
First, we study Bayesian search problems under full revelation, in which the DM observes all sample values. 
The monotonicity condition under full revelation is extremely restrictive: it requires both a single-crossing order and an endpoint restriction on the crossing point (Theorem~\ref{thm: IFF full}).  Moreover, monotonicity under full revelation implies monotonicity under lower censoring.  A key insight is that informational coarsening is not merely a feature of the model but a necessary structural ingredient for preserving monotonicity. 
Thus, information censoring serves as a stabilizing force in sequential decision-making under uncertainty. 
%Among all the information structures that fully reveal the better-off sample outcomes, i.e., all the sample outcomes satisfying $X>y$ are fully revealed, lower-censoring pools all the non-improving outcomes into a single signal, and therefore is the least informative information structure and renders the weakest condition for monotonicity.\EL{The last sentence is not clear to me.} 

Second, we extend our analysis to expected monotonicity, which requires the conditional expectation of the one-step marginal search value to decrease over time. 
%Under the single-crossing order, we characterize the necessary and sufficient condition for expected monotonicity. 
We show that the condition for expected monotonicity is equivalent to the conditional-expectation version of the monotonicity condition under lower censoring (Theorem~\ref{thm: sm monotone}). 
A striking implication is that the expected monotonicity condition is identical under lower censoring and full revelation. 
%Although these two information structures induce different distributions of posterior beliefs, they generate the same ex ante expected posterior shift. 
This is because $U(y,\alpha)$ is affine in $\alpha$, and therefore, the expected learning effect is invariant across different information structures, while the expected fallback-value effect remains unchanged.  The condition for expected monotonicity is weaker, yet it is insufficient to guarantee the optimality of myopic, single-threshold cutoff rules.

While Theorems \ref{thm: monotone conditions_lower}, \ref{thm: IFF full}, and \ref{thm: sm monotone} derive sharp characterizations within the binary-distribution framework, the dynamic programming approach applies more broadly, establishing the existence of the minimal fixed point in general multi-distribution settings and its role in determining an optimal stopping rule. In this more general environment, however, monotonicity need not hold.

The paper contributes to three strands of the literature. 
First, it advances the literature on search with learning by identifying primitive conditions on information structures under which Bayesian updating preserves monotonicity. 
Whereas existing work typically imposes monotonicity directly on posterior predictive distributions or relies on tractable conjugate priors, we show that signal censoring itself disciplines posterior dynamics and restores reservation-cutoff policies. 
Second, it contributes to the analysis of censored information by characterizing how censored likelihood ratios enter posterior odds and continuation values in sequential search. 
Third, while related to sequential experimentation and multi-armed bandits, our mechanism differs fundamentally due to the endogeneity of the stopping payoff: under free recall, each draw simultaneously updates beliefs and potentially upgrades the fallback option. 
The interaction between belief updating and this evolving reservation payoff forms the core of our analysis.

The remainder of the paper is organized as follows. 
Section~\ref{sec: literature} reviews the related literature. 
Section~\ref{sec: model} introduces the baseline Bayesian search model and defines monotone search problems. 
Section~\ref{sec: lower-censoring} characterizes monotonicity and the optimal stopping rule under lower censoring. 
Section~\ref{sec: job search} and \ref{sec: consumer search} present two applications of our main results. 
Section~\ref{sec: further discussions} develops further extensions of the baseline model.  Section~\ref{sec: conclusion} concludes. 
All proofs are relegated to the Appendix.

\section{Related Literature}\label{sec: literature}
This paper contributes to the literature on sequential search with learning.  In classical search problems with free recall and a known payoff distribution, it is known that the optimal stopping rule is myopic and takes the form of reservation cutoffs \citep{mccall1970economics,weitzman1979optimal}.  The search problem becomes substantially more difficult when the payoff distribution becomes unknown, because each observation affects both the current fallback value and the posterior predictive distribution of future draws.  The central question is whether the monotonicity and reservation-cutoff properties can be preserved when Bayesian learning is introduced.

\paragraph{Search with learning.}
In an early paper,  \citet{rosenfield1981optimal} study Bayesian sequential search with free recall and no recall.  In particular, for free recall, assuming the one-step marginal search value never increases with additional samples, they show that the optimal stopping rule is myopic and takes the form of reservation cutoffs.  They further suggest examples that satisfy, e.g., Dirichlet-multinomial model, or do not satisfy the assumption, e.g., a normal distribution with an unknown mean. \citet{bikhchandani1996optimal} study optimal search with an unknown payoff distribution.  They suggest two assumptions on the posterior update process under which the search is monotone.  Their Assumption~1 requires that, after observing an additional sample value below the fallback value $y$, the posterior probability that the next sample value exceeds $y$ decreases.  Their Assumption~2 requires the posterior probability to depend only on the number of previous observations, not on the realized search history.  Under the assumptions, the optimal stopping rule takes a simple reservation form, and the reservation cutoffs are history-independent.  The assumptions are satisfied in an \emph{ad hoc} learning model, in which the posterior predictive distribution is a convex combination of the prior and the empirical distribution, as well as the Dirichlet-multinomial model.  \citet{talmain1992search} provides a closed-form solution for the optimal search rule in the Dirichlet-multinomial case.\footnote{The tractability of the Dirichlet-multinomial model has made it useful in empirical and applied search models with Bayesian learning; see, for example, \citet{koulayev2013search}, \citet{santos2017search}, \citet{ulu2009uncertainty}, and
\citet{smith2017risk}.}

\cite{adam2001learning} extends the above problem to multiple unknown distributions.  Specifically, a DM searches for a prize across a set of distinguishable and independent channels.  For each channel, there is an unknown distribution, and she can update her belief based on previous search outcomes of this channel.  He imposes a strong assumption (A1) that the reservation prize of each channel is decreasing over time.  This assumption rules out the possibility that a strong positive learning effect can dominate the negative fallback value effect, which makes the problem monotone.  Under the assumption, he shows that a variant of Pandora's rule \emph{$\grave{a}$ la} \citet{weitzman1979optimal} is optimal. %and the reservation cutoffs are history-dependent.

Our paper differs from this line of work in the source of monotonicity.  Rather than imposing monotonicity restrictions directly on the posterior predictive process, we derive monotonicity from a primitive information structure, e.g., lower-censoring.  Bayesian updating therefore uses censored likelihood terms such as $F_i(y)$.  Assuming a single-crossing order between the candidate distributions, a censored observation drives the DM's posterior belief to move towards the inferior distribution.  This
allows the negative fallback-value effect to dominate the learning effect under mild conditions, and restores the monotonicity that generally fails under unrestricted Bayesian learning.

A broader Bayesian search literature studies optimal stopping when the unknown distribution
belongs to a parametric family.  \citet{degroot1968some} is an early contribution that studies
Bayesian sequential search with normally distributed payoffs and an unknown mean.  More
recently, \citet{baucells2025search} analyze search with recall when both the mean and variance
of a normal distribution are unknown, and beliefs are updated with a conjugate prior.  They show
that the classical single-reservation-price structure generally fails and derive an alternative
optimal stopping rule based on standardized state variables.  \citet{preuss2023search} analyzes
search, learning, and tracking in markets where consumers may learn about both prices and their
own willingness to pay.  These papers reinforce the common theme that Bayesian learning often
destroys the simple cutoff structure of classical search unless additional structure is imposed.

\paragraph{Censored information and Bayesian updating.}
The paper is also related to work on learning from censored data.  In Bayesian statistics,
censored observations are informative signals and enter posterior updating through censored
likelihood terms.  \citet{ferguson1979bayesian} study Bayesian
learning with censored data and show that certain prior families preserve tractability under
censoring.  Our model uses this statistical insight in a sequential-search environment, e.g., a
lower-censored event such as $\{X\le y\}$ changes posterior odds by the Bayes factor
$F_1(y)/F_2(y)$.

A related economics literature studies endogenous censoring generated by stopping behavior.
\citet{he2022mislearning} studies mis-learning from endogenously censored histories in optimal
stopping problems.  Because predecessors stop once a sufficiently good realization occurs,
observed histories are selected, and behavioral misperceptions such as the gambler's fallacy
can lead to systematically biased beliefs and premature stopping.  Our model is different in
two respects.  First, the DM is fully Bayesian and correctly accounts for the censoring rule.
Second, censoring occurs within the DM's own search process: the current fallback value
determines which outcomes are fully revealed and which are observed only coarsely.

There is also a related machine-learning literature on Bayesian optimization and active
learning with censored responses.  For example, \citet{hutter2013bayesian} study Bayesian
optimization when evaluations may be right-censored.  That literature is methodologically
related because censored observations are incorporated into posterior or acquisition rules,
but the economic problem is different.  In our setting, the censoring threshold is endogenous:
it is determined by the current fallback value in a costly search problem with free recall.

\paragraph{Relation to bandits and experimentation.}
While our model shares the stop-or-continue structure of a sequential learning problem, it departs fundamentally from the canonical multi-armed or one-armed bandit framework \citep{gittins1974dynamic,bergemann2018bandit}. In a standard bandit setting, the DM repeatedly decides whether to experiment with a risky arm that yields flow rewards or to take an exogenous safe alternative. By contrast, the outside option is entirely endogenous in a search problem, and each draw affects the state through two distinct channels: it provides information about the unknown distribution, and it may raise the fallback value.  Our monotonicity analysis is built precisely around the interaction between these two channels; when coupled with one-sided censoring, it is the structural persistence of this recall-driven fallback effect that systematically tethers and dampens positive belief shifts, preserving global monotonicity where full revelation fails.

The distinction is also important relative to recent work that combines search and experimentation.
\citet{fershtman2025searching} study experimentation with endogenous consideration
sets, where the DM alternates between experimenting with alternatives already in the consideration
set and searching for new alternatives.  Such models are closer to multi-armed experimentation:
the set of available arms evolves over time, and the DM allocates attention across alternatives.
By contrast, our model isolates learning from a single population of \emph{ex ante} homogeneous
boxes and studies how endogenous censoring of search outcomes affects the structure of optimal
stopping.  Recent work on Bayesian search without recall, such as \citet{xu2025search}, is also
distinct: without recall, the payoff-relevant state is not the best previously discovered value,
whereas our model's key mechanism is precisely the endogenous fallback value generated by free
recall.

\section{The Model}\label{sec: model}
\subsection{Model Setup}\label{subsec: model setup}
Consider a Bayesian decision maker (DM) who searches sequentially with free recall through a countably infinite set of \emph{ex ante} identical boxes. Each box contains a hidden reward $X \in \mathbb{R}_+$. The DM does not know the true underlying distribution $L$, but knows it belongs to a finite set of candidates $\{L_1, \ldots, L_\ell\}$. Conditional on $L = L_i$, the box rewards are independently and identically distributed (i.i.d.) draws from $L_i$. The DM's prior belief over the candidate distributions is given by $\boldsymbol{\alpha}_0 = (\alpha_{0,1}, \ldots, \alpha_{0,\ell}) \in \Delta^{\ell-1}$, where $\alpha_{0,i} = \Pr(L = L_i)$. In each period $t \ge 1$, the DM decides whether to terminate the search or incur a constant marginal cost $c > 0$ to inspect an additional box. 

Let $X_t$ denote the latent reward of the box inspected at time $t$. Under the prevailing information structure, the DM may not observe the exact realization of $X_t$, but instead observes a signal generated by it. The DM uses this signal to update her posterior belief regarding $L$. Despite this partial observability, the information structure guarantees that the DM can perfectly track the current \emph{fallback value}. Initialized at $Y_0 = 0$, the fallback value evolves recursively according to $Y_t = \max\{Y_{t-1}, X_t\}$, representing the highest reward discovered up to step $t$. Upon terminating the search at time $t$, the DM claims $Y_t$.

Let $(\Omega,\mathcal{F},\mathbb{P})$ be a probability space supporting the random measure $L$ and the sequence of latent rewards $\mathbf{X}=\{X_t\}_{t=1}^{\infty}$. The observed signals generate a natural filtration $\mathbb{F}=\{\mathcal{F}_t\}_{t=0}^{\infty}$, where $\mathcal{F}_0$ is the trivial $\sigma$-algebra. A stopping time $\tau$ is \textit{admissible} if it is based only on available information, namely, $\{\tau=t\}\in\mathcal{F}_t$ for every $t$.
The sequential search problem can be formulated as an optimal stopping problem. The DM chooses an admissible stopping time to maximize the expected net payoff:
\begin{equation}\label{eq:optimal_stopping_problem}    \sup_{\tau\in\mathcal{T}}\mathbb{E}_{\boldsymbol{\alpha}_0}\left[Y_\tau-c\tau\right],
\end{equation}
where $\mathbb{E}_{\boldsymbol{\alpha}_0}$ denotes expectation with respect to the probability measure induced by the prior belief $\boldsymbol{\alpha}_0$.
%%%%%%%%%%%%%%%%%%%%%%%%%%%%%%%%%%%%%%%%%%%%

An alternative formulation of the decision problem, rather than expressing it directly in terms of the optimal stopping rule in \eqref{eq:optimal_stopping_problem}, is to characterize the value of continuation as a function of the current state. Let
$
J^r(y,\boldsymbol{\alpha})
$
denote the value function when the current fallback value is $y$, and the posterior belief over the underlying distributions is $\boldsymbol{\alpha}$. The superscript $r$ indicates the information structure under which the posterior is updated. The corresponding Bellman equation is (see \cite{rosenfield1981optimal}).
\begin{equation}\label{eq:Bellman-general}
    J^r(y,\boldsymbol{\alpha})
    =
    \max\left\{
        y, \,
        \E_{\boldsymbol{\alpha}}\left[
            J^r\bigl(\max\{y, X\},B(\boldsymbol{\alpha},X)\bigr)
        \right]-c
    \right\},
\end{equation}
where $X$ is drawn from the posterior predictive distribution
$L_{\boldsymbol{\alpha}}$, and $B(\boldsymbol{\alpha},X)$ is the posterior belief induced by the information available to the  DM upon the realization of $X$. Thus, the updating rule depends on the information structure.

The first term in the maximization corresponds to stopping immediately and obtaining the best payoff currently available, namely $y$. The second term is the continuation value: the DM pays the search cost $c$, observes a realization $X$ drawn from the posterior predictive distribution, updates the posterior belief, and updates the fallback value.

%%%%%%%%%%%%%%%%%%%%%%%%%%%%%%%%%

We focus on the information structure of lower-censoring.\footnote{
\cite{bikhchandani1996optimal} study a search problem with a different form of censored data. In their model, the DM either observes the true realization or receives a censored signal indicating that the realization lies above or below a given, exogenously specified threshold. Our setup differs in at least two respects. First, we study Bayesian learning, whereas they consider an \emph{ad hoc} learning rule. Second, the censoring threshold in our model is endogenous: an observation is censored only when it is weakly below or above the current fallback value. Consequently, not every sequence of uncensored and censored observations that is feasible in their formulation can arise in our model.  The alternative structures of upper-censoring and full revelation are analyzed in Sections \ref{sec: consumer search} and \ref{subsec: full revelation}.}
To be specific, upon inspecting the $(t+1)$-th box with a realized value $X_{t+1}=x$, the DM observes the signal:
\begin{equation}\label{eq:lower-censored-signal}
S_{t+1}=
\begin{cases}
\{X_{t+1}\le Y_t\}, & \text{if } x\le Y_t,\\[0.1cm]
x, & \text{if } x>Y_t.
\end{cases}
\end{equation}
Under lower-censoring, realizations that fail to strictly improve upon the current fallback value are censored, pooling into a single signal $\{X_{t+1}\le Y_t\}$, whereas strictly improving realizations are fully revealed. This lower-censoring structure captures environments in which unfavorable outcomes generate only coarse information, while favorable outcomes are observed more precisely.

Let $h_t := (S_1, \ldots, S_t)$ denote the realized search history after $t$ inspections, where $h_0$ denotes the empty history. Given history $h_t$, the DM's posterior belief over the set of candidate distributions is:
\[
    \boldsymbol{\alpha}_t 
    = \boldsymbol{\alpha}(h_t) 
    = (\alpha_{t,1}, \ldots, \alpha_{t,\ell}) \in \Delta^{\ell-1}, 
    \quad \text{where} \quad 
    \alpha_{t,i} := \Prob(L = L_i \mid h_t).
\]

\paragraph{The Markov Decision Problem.}
The sequential search problem can be formulated as a Markov decision process (MDP). Because the search allows for free recall, the state of the system at any period $t$ is fully characterized by the sufficient statistic $(Y_t, \boldsymbol{\alpha}_t)$, representing the current fallback value and the posterior belief, respectively.

Suppose the process is in state $(y, \boldsymbol{\alpha})$ and the DM chooses to inspect an additional box. Let $X$ be the latent reward of this draw and $S$ be the realized signal generated according to the lower-censoring rule \eqref{eq:lower-censored-signal}. The posterior belief updates via Bayes' rule, which we denote by $\boldsymbol{\alpha}' = B(y, \boldsymbol{\alpha}, S)$. The state transitions to $(y', \boldsymbol{\alpha}')$, where:
\begin{equation}\label{eq: updated state}
    y' = \max\{y, X\}, \text{ and } \boldsymbol{\alpha}' = B(y, \boldsymbol{\alpha}, S).
\end{equation}
\ignore{
Let $J(y,\boldsymbol{\alpha})$ denote the value function at state $(y,\boldsymbol{\alpha})$, defined as the supremum of expected continuation profits, excluding sunk search costs. Then $J$ satisfies the Bellman equation:
\begin{equation}\label{eq: Bellman_general}
    J(y, \boldsymbol{\alpha}) = \max\left\{ y, \, \E_{\boldsymbol{\alpha}}\left[ J\bigl(\max\{y, X\}, B(y, \boldsymbol{\alpha}, S)\bigr) \right] - c \right\},
\end{equation}
where $\E_{\boldsymbol{\alpha}}$ denotes the expectation with respect to the posterior predictive distribution of the latent reward, $X \sim L_{\boldsymbol{\alpha}} := \sum_{i=1}^\ell \alpha_i L_i$, with the signal $S$ being deterministically induced by $X$ and $y$ under the signaling scheme.

Equation \eqref{eq: Bellman_general} characterizes a fixed point of the associated Bellman operator. The formal arguments establishing the existence and uniqueness of the solution $J$ are deferred to Section \ref{sec: general}.
}
 We next propose some situations that fit this structure.

\heading{Job search with screening}  
A job seeker searches sequentially for outside employment opportunities, with each application incurring a cost $c$. She is currently employed in a job with value $y$, and can always remain in her current position or accept any offer she receives. Her match value $X_i$ with firm $i$ is an i.i.d.\ draw from an unknown distribution $L$, and she updates her belief about $L$ over time based on the outcomes of her applications. Applications are screened by recruiters or hiring platforms. If an application is rejected, she typically learns only that her match value fell below a screening threshold. By contrast, if she is invited to interview or receives an offer, she obtains substantially more precise information about her match value with the employer. This setting naturally fits the framework of sequential search with lower censoring.

\heading{Product development} The same logic applies to product development. A firm may experiment with successive versions of a product, using the current best version as a benchmark for future improvements. If a new version performs worse than the current benchmark, detailed performance metrics may not be collected or analyzed extensively. The firm then learns only that the new version failed to improve upon the existing product. In contrast, if the new version outperforms the benchmark, the firm has stronger incentives to measure and analyze its performance more carefully, thereby obtaining more precise information about its value. The resulting information structure is one in which inferior outcomes are lower-censored, whereas improvements are fully observed.

\heading{Consumer price search} The lower-censoring model has a natural analogue in minimization problems with upper-censoring. Consider, for example, a consumer who searches sequentially for a lower price. The consumer is uncertain about the distribution of prices and updates her beliefs based on search outcomes. Search is conducted with free recall, and the current best price is $p$. Prices above $p$ are unattractive, so the consumer may not record their exact values or investigate them further, learning only that they exceed the current best price. By contrast, a price below $p$ represents a valuable discovery and is observed precisely. Repeated observations of prices exceeding $p$ may make the consumer increasingly pessimistic about the distribution of market prices and eventually induce her to stop searching and purchase at the current best price $p$. This setting is the upper-censoring counterpart of the lower-censoring environment studied in this paper.

\subsection{Monotone Search Problems}
In general, the search problem \eqref{eq:Bellman-general}  does not admit a simple closed-form solution. This section identifies a class of monotone problems for which the optimal stopping rule is myopic and is in the form of state-dependent cutoff values.  For any given state $(y, \boldsymbol{\alpha})$, define the one-step marginal search value by:
\begin{equation}\label{eq:U-general}
    U(y,\boldsymbol\alpha)
    :=
    \E_{\boldsymbol\alpha}\bigl[\max\{y,X\}-y\bigr]
    =
    \E_{\boldsymbol\alpha}\bigl[(X-y)_+\bigr].
\end{equation}
This function quantifies the expected gross increment to the current fallback value derived from sampling exactly one more box. Evaluated along a realized search history $h_t$, we denote this state-dependent marginal benefit by $U_t := U(Y_t, \boldsymbol{\alpha}_t)$.

\begin{definition}[Monotone Search Problem]\label{def: Monotone Search Problem}
The search problem \eqref{eq:Bellman-general} is \emph{monotone} if, for every state $(y,\boldsymbol{\alpha})$ and every subsequent state $(y',\boldsymbol{\alpha}')$ that can arise from a feasible signal realization,
\begin{equation}\label{def: monotone}
    U(y,\boldsymbol{\alpha}) \ge U(y',\boldsymbol{\alpha}').
\end{equation}
\end{definition}

The monotonicity condition \eqref{def: monotone} is \emph{pathwise}: after any possible continuation outcome, the one-step marginal value of search cannot increase. Consequently, once a one-step search is not worthwhile, it remains unprofitable after every possible future realization.\footnote{Pathwise monotonicity is stronger than expected monotonicity, i.e., $U(y,\boldsymbol{\alpha})\ge\E_{\alpha}[U(y',\boldsymbol{\alpha}')]$. However, it is particularly convenient because it guarantees that once the myopic stopping rule recommends stopping, it continues to do so after every possible future realization. We return to expected monotonicity in Section~\ref{subsec: expected monotonicity}.} The following proposition establishes the optimal stopping rule for monotone search problems.\footnote{Its proof, as well as the proofs of all subsequent results, is deferred to the appendix.} 

\begin{proposition}[Optimal stopping for monotone search problems]
\label{prop: optimal stopping monotone}
Suppose that condition~\eqref{def: monotone} holds. Then the stopping time
\begin{equation*}
    \tau^*:=\inf\{t\ge 0: U(Y_t,\boldsymbol{\alpha}_t)\le c\}
\end{equation*}
is optimal. Equivalently, at any state $(y,\boldsymbol{\alpha})$, it is optimal to stop if and only if
$
U(y,\boldsymbol{\alpha})\le c.
$
\end{proposition}
The intuition for Proposition \ref{prop: optimal stopping monotone} relies on the structural persistence imposed by pathwise monotonicity. A myopic rule dictates stopping when the one-step expected marginal benefit falls below the marginal cost, namely $ U(Y_t,\boldsymbol{\alpha}_t) \le c$. Without pathwise monotonicity, a DM might rationally continue searching even when $ U(Y_t,\boldsymbol{\alpha}_t) \le c$, treating the immediate expected loss as an investment to reach future states where the one-step marginal search value recovers and exceeds $c$. Condition~\eqref {def: monotone} strictly precludes this speculative continuation. Because $U(y,\alpha)$ is weakly decreasing along every feasible realized path, reaching a boundary state where $U(y,\alpha) \le c$ guarantees that $U(y',\alpha') \le c$ for all subsequent reachable states $(y',\alpha')$. Consequently, the expected net gain of any arbitrary multi-step continuation strategy can be decomposed into a sequence of one-step expected net gains, all of which are bounded above by zero. Because the one-step marginal search value can never recover once it falls below the cost threshold, no multi-step strategy can strictly improve upon an immediate stop.

In classical search problems where the distribution $L$ is known, the search problem \eqref{eq:Bellman-general} is trivially monotone. To see this, note that $U_t=\mathbb{E}[(X-Y_t)_+]$ does not depend on the belief and is decreasing in the fallback value $Y_t$, which weakly increases over time under free recall. Hence, $U_t\ge U_{t+1}$ for every $t$. 

When $L$ is unknown, however, the DM updates her beliefs over time, and these belief updates may either increase or decrease the one-step marginal value of search. Consequently, the search problem is not monotone in general.

%%%%%%%%%%%%%%%%%%%%%%%%%%%%%%%%%%%% 11111%%%%%%%%%%%%%%%%%%%%%%%

\subsection{A Simple Learning Model}
This section specializes the model to a finite-support environment with two distributions. Specifically, suppose that $X$ takes values in the finite set $K:=\{1,2,\ldots,k\}$, and let $\bar K:=K\cup\{0\}$ denote the set of fallback values. We assume that $\ell=2$, so that $L\in\{L_1,L_2\}$, and write
\[
    L_1=(\pi_{1,1},\ldots,\pi_{1,k}),\quad\text{ and }\quad
    L_2=(\pi_{2,1},\ldots,\pi_{2,k}),
\]
with $\pi_{i,j}>0$ and $\sum_{j=1}^k\pi_{i,j}=1$ for $i=1,2$.  We assume $\pi_{1,k}<\pi_{2,k}$ for simplicity.\footnote{It simplifies the discussion of some boundary cases, yet is not essential for the main results.} For $i=1,2$, denote $F_i(x):=\sum_{j=1}^x\pi_{i,j}$ with $F_i(0):=0$ as the distribution functions. The posterior belief after a search history $h_t$ is $\alpha_t:=\Prob(L=L_1\mid h_t)$.  Given a posterior belief $\alpha\in[0,1]$, the posterior predictive distribution is $L_\alpha =\alpha L_1+(1-\alpha)L_2$, with its mass and distribution functions being
\begin{equation*}
    \pi_\alpha(x):=\alpha \pi_{1,x}+(1-\alpha)\pi_{2,x}
    \quad\text{ and }\quad
    F_\alpha(x):=\alpha F_1(x)+(1-\alpha)F_2(x).
\end{equation*}
Define the survival functions as follows
\begin{equation}
\label{eq:survival}
\bar F_\alpha(x):=1-F_\alpha(x), 
\qquad
\bar F_i(x):=1-F_i(x), \quad i=1,2.
\end{equation}
We next introduce several stochastic orders that formalize different notions of $L_1$ being smaller than $L_2$.
\begin{definition}[Stochastic orders]\label{def: st orders}
We say $L_1$ is smaller than $L_2$  %(i) in likelihood-ratio order, denoted $L_1\prec_{lr}L_2$, if $\pi_{1,x}/\pi_{2,x}$ is weakly decreasing in $x$; 
(i) in single-crossing order, denoted $L_1\prec_{sc}L_2$, if the sequence $(\pi_{1,x}-\pi_{2,x})_{x=1}^k$ changes sign from positive to negative at most once as $x$ increases; (ii) in first-order stochastic dominance (FOSD) order, denoted $L_1\prec_{st}L_2$, if $F_1(x)\ge F_2(x)$ for all $x\in K$; and (iii) in increasing-convex order, denoted $L_1\prec_{icx}L_2$, if for all $y\in\bar{K}$,
\[
\mathbb E_{L_1}[(X-y)_+]\le \mathbb E_{L_2}[(X-y)_+].
\]
%Equivalently,  $\Delta(y):=\sum_{z=y}^k (\pi_{2,z}-\pi_{1,z})(z-y)\ge 0$ for all $y\in \bar{K}$.
\end{definition}

For distributions defined on the ordered set $K$, it is known that the stochastic orders satisfy
\begin{equation}\label{eq: st order relations}
   % L_1\prec_{lr}L_2     \quad\Longrightarrow\quad
    L_1\prec_{sc}L_2
    \quad\Longrightarrow\quad
    L_1\prec_{st}L_2
    \quad\Longrightarrow\quad
    L_1\prec_{icx}L_2 . 
\end{equation}   
The lemma below records a useful monotonicity property of posterior predictive distributions. 

\begin{lemma}\label{lem: st order}
Suppose $L_1\prec_o L_2$ for some $o\in\{sc, st, icx\}$. Then, for any $\alpha,\alpha'\in[0,1]$, 
\[
L_\alpha\prec_o L_{\alpha'}\quad\Longleftrightarrow\quad \alpha\ge \alpha'.
\]
\end{lemma}

\section{Bayesian Search under lower-censoring}\label{sec: lower-censoring}
\subsection{Characterization of monotonicity under lower-censoring} 
This section studies the search problem \eqref{eq:Bellman-general} under the information structure of lower-censoring; see \eqref{eq:lower-censored-signal}.
To this end, we first derive several properties of the one-step marginal search value. In the learning model with two candidate distributions, the one-step marginal search value at state $(y,\alpha)$ is
\begin{equation*}
U\left( y,\alpha \right)
:=\mathbb{E}_{\alpha}\!\left[\left(X-y\right)_+\right]
=\sum_{z=y+1}^{k}\pi_{\alpha}(z)(z-y)
=\alpha U_{1}(y)+(1-\alpha)U_{2}(y),
\end{equation*}
where
$
U_i(y):=\mathbb{E}_{L_i}\!\left[(X-y)_+\right], i=1,2,
$
denotes the one-step marginal search value under distribution $L_i$. Rearranging the terms and interchanging the order of summation yields the following equivalent representation of $U(y,\alpha)$ in terms of the survival functions.
\begin{equation}\label{eq: U(y,a) survival}
U(y,\alpha)
=
\sum_{z=y+1}^k \pi_\alpha(z)\left(\sum_{m=y}^{z-1} 1\right)
=\sum_{m=y}^{k-1}\left(\sum_{z=m+1}^k \pi_\alpha(z)\right)
=\sum_{m=y}^{k-1} \bar F_\alpha(m).
\end{equation}
We further define  the difference between $U_2(y)$ and $U_1(y)$ as
\begin{equation}\label{eq: Delta(y)}
\Delta \left( y\right) :=U_{2}\left( y\right) -U_{1}\left( y\right)
=\sum_{z=y+1}^{k}\left( \pi_{2,z}-\pi_{1,z}\right) \left( z-y\right) .
\end{equation}
Given a state $(y,\alpha)$, the updated posterior under lower-censoring is then
\begin{equation}\label{eq: a_t+1 lower}
    \alpha'=
    \begin{cases}
        \alpha'_y:=\dfrac{\alpha F_1(y)}{F_{\alpha}(y)},
        & \text{if }\{X\le y\}\text{ is observed},\\[1.25em]
        \alpha'_x:=\dfrac{\alpha \pi_{1,x}}{\pi_{\alpha}(x)},
        & \text{if } X=x>y\text{ is observed}.
    \end{cases}
\end{equation}

%%%%%%%%%%%%%%%%%%%%%%%%

\begin{assumption}[Single crossing]\label{ass: single-crossing}
 $L_{1}\prec _{sc}L_{2}$.
\end{assumption}

We introduce the single-crossing assumption, and define the following crossing value
\begin{equation}\label{eq: x_hat}
    \hat{x}:=\max\{x\in K:\pi_{1,x}-\pi_{2,x}\ge 0\}.
\end{equation}
Under Assumption \ref{ass: single-crossing} and that $\pi_{1,k}<\pi_{2,k}$, we have $\pi_{1,x}-\pi_{2,x}\ge 0$ for $x\le \hat{x}$ and $\pi_{1,x}-\pi_{2,x}<0$ for $x>\hat{x}$.

The single-crossing order turns out to be useful in the study of search problems with Bayesian learning. First, by \eqref{eq: st order relations}, it implies increasing convex order, i.e., $\Delta(y)=U_2(y)-U_1(y)\ge 0$ by definition.  Therefore, the one-step marginal search value under $L_2$ is greater than that under $L_1$, suggesting that $L_1$ is the inferior distribution.  Second, it implies FOSD, i.e., $F_1(x)\ge F_2(x)$, and hence, a censored observation $\{X\le y\}$ drives the updated posterior towards the inferior distribution $F_1$ by \eqref{eq: a_t+1 lower}.  Third, the existence of the crossing value $\hat{x}$ also indicates that, if a fully revealed value $X=x>\hat{x}$, the updated posterior will move towards the superior distribution $L_2$ by \eqref{eq: a_t+1 lower}; if $x\le\hat{x}$, the updated posterior will move weakly towards the inferior distribution $L_1$.  In short, the single-crossing order not only unambiguously ranks the two distributions in terms of their marginal search values, but also provides a clear pattern on the directions of Bayesian belief shifts under one-sided censoring.

Lemmata~\ref{lem: Delta(y)} and~\ref{lem: U(y,a)} collect the basic properties of $U(y,\alpha)$ and $\Delta(y)$ used in the subsequent analysis.
\begin{lemma}\label{lem: Delta(y)}
Suppose Assumption \ref{ass: single-crossing} holds.  For all $y\in\bar{K}$, (i) $F_1(y)-F_2(y)\geq 0$; (ii) $\Delta(y)\ge 0$ with $\Delta(k)=0$; and (iii) $\Delta(y)$ is weakly decreasing in $y$. 
\end{lemma}

 Lemma \ref{lem: Delta(y)} shows that under Assumption \ref{ass: single-crossing}, $F_1(y)\ge F_2(y)$ for all $y\in \bar{K}$.  It follows from \eqref{eq: a_t+1 lower} that, for $\alpha\in(0,1)$, the sign of the posterior shift under lower-censoring is 
\begin{equation}\label{eq: a'-a_lower_sign}
\alpha'-\alpha\propto
\begin{cases}
\displaystyle
F_1(y)-F_2(y),
& \text{if } \{X\le y\} \text{ is observed},\\[0.5em]
\displaystyle
\pi_{1,x}-\pi_{2,x},
& \text{if } X=x>y \text{ is observed}.
\end{cases}
\end{equation}
Therefore, a censored observation $\{X\le y\}$ increases the posterior belief and makes the DM more pessimistic as the posterior of the inferior distribution $L_1$ increases.  The lemma also shows the one-step marginal search value under $L_2$ is always greater than that under $L_1$, as $\Delta(y)\ge 0$.

\begin{lemma}\label{lem: U(y,a)}
(i) For each $\alpha\in[0,1]$, the function $y\mapsto U\left( y,\alpha \right) $ is strictly decreasing on $y\in\{0,\ldots, k-1\}$, and is discretely convex on $y\in\bar{K}$; and \\ (ii) under Assumption \ref{ass: single-crossing}, $U\left( y,\alpha \right) $ is decreasing in $\alpha $ and has increasing differences in $\left( y,\alpha \right) $. That is, for $y'>y$ and $\alpha'>\alpha$, $\bigl[U(y',\alpha')-U(y,\alpha')\bigr]-\bigl[U(y',\alpha)-U(y,\alpha)\bigr]\ge 0$.
\end{lemma}

\begin{figure}[t]
\centering
\begin{tikzpicture}
\begin{axis}[
    width=12cm,
    height=10cm,
    xlabel={Fallback Value},
    ylabel={$U(y,\alpha)$},
    xmin=0,
    xmax=10,
    ymin=0,
    xtick={0,3,5,6,10},
    xticklabels={$0$,$y$,$\hat{x}$,$y'$,$k$},
    ytick=\empty,
    legend style={draw=none, fill=none, at={(0.95,0.95)}, anchor=north east},
    legend cell align={left},
]
% Curve for alpha=0 (no marks)
\addplot+[mark=none, thick, dashed] coordinates {(0,7.150000) (1,6.150000) (2,5.160000) (3,4.200000) (4,3.290000) (5,2.450000) (6,1.700000) (7,1.060000) (8,0.550000) (9,0.190000) (10,0.000000)};
\addlegendentry{$U(\cdot,\alpha')$: $\alpha'<\alpha$}

% Curve for alpha=0.5 (no marks)
\addplot+[mark=none, thick, dashed] coordinates {(0,5.912500) (1,4.912500) (2,3.990000) (3,3.150000) (4,2.397500) (5,1.737500) (6,1.175000) (7,0.715000) (8,0.362500) (9,0.122500) (10,0.000000)};
\addlegendentry{$U(\cdot,\alpha)$}

% Curve for alpha=1 (no marks)
\addplot+[mark=none, thick, dashed] coordinates {(0,4.675000) (1,3.675000) (2,2.820000) (3,2.100000) (4,1.505000) (5,1.025000) (6,0.650000) (7,0.370000) (8,0.175000) (9,0.055000) (10,0.000000)};
\addlegendentry{$U(\cdot,\alpha')$: $\alpha'>\alpha$}

% Mark the three specified points
\draw[fill=blue] (axis cs:3,3.150000) circle[radius=2pt];
\draw[fill=blue] (axis cs:6,1.175000) circle[radius=2pt];
\draw[fill=blue] (axis cs:6,1.700000) circle[radius=2pt];

% Vertical dashed line at x=5 (\hat{x})
\draw[dashed, gray!70] (axis cs:5,0) -- (axis cs:5,7.5);

% Arrow along the alpha=0.5 curve (fallback value effect)
\draw[->, very thick, blue] 
    (axis cs:3,3.150000) -- 
    (axis cs:4,2.397500) -- 
    (axis cs:5,1.737500) -- 
    (axis cs:6,1.175000);

% Vertical arrow (learning effect)
\draw[->, very thick, red] (axis cs:6,1.175000) -- (axis cs:6,1.700000);

% ----- Labels with S‑shaped curved arrows pointing to the respective midpoints -----

% Southwest label for "fallback value effect"
% Midpoint of the blue arrow is approximately (4.5, 2.07)
\node[blue, anchor=north east, align=right] (fallback_label) at (axis cs:4,1.4) {B: fallback value effect};
\draw[<-, blue] 
    (fallback_label.east) .. controls +(1.2,0.6) and +(-0.6,-0.6) .. (axis cs:4.5,2.0);

% Northeast label for "learning effect"
% Midpoint of the red arrow is (6, 1.4375)
\node[red, anchor=south west, align=left] (learning_label) at (axis cs:6.5,3) {A: learning effect};
\draw[<-, red]
    (learning_label.south) .. controls +(0,-1.5) and +(1.5,1) .. (axis cs:6.05,1.4375);

\end{axis}
\end{tikzpicture}
\caption{Inter-temporal Decomposition of $U(y',\alpha')-U(y,\alpha)$}\label{fig: U(y,a) decomposition}
\floatfoot{In this example, $k=10$; for $L_1$, $\pi_{1,1}=0.145$ and $\pi_{1, j+1}=\pi_{1,j}-0.01$; for $L_2$, $\pi_{2,1}=0.01$ and $\pi_{2, j+1}=\pi_{2,j}+0.02$.}
\end{figure}
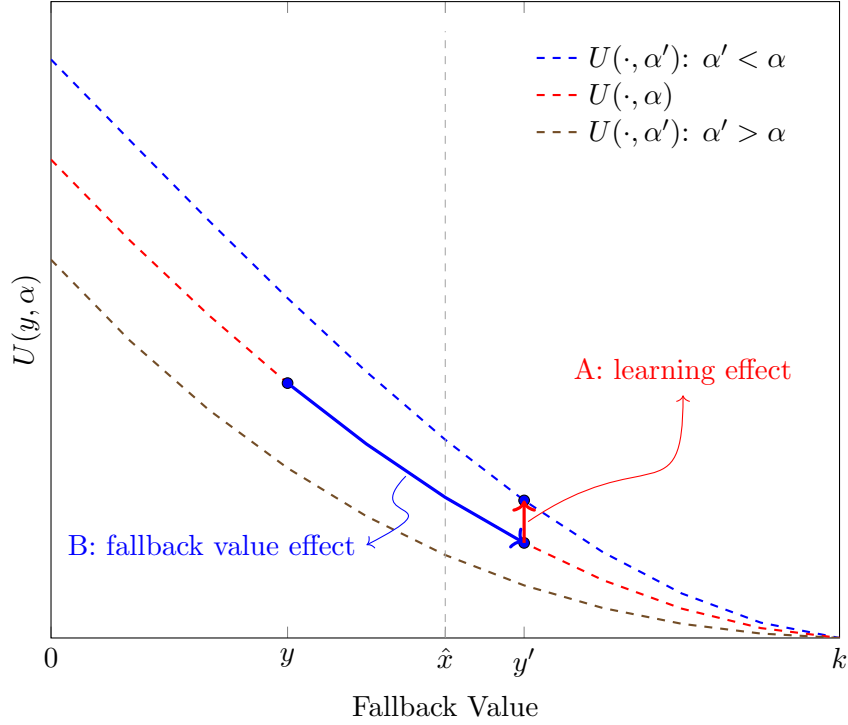

We decompose the intertemporal change of $U(y,\alpha)$ into a fallback-value effect and a learning effect:
\begin{equation}\label{eq: U(y',a')-U(y,a)}
U\left( y',\alpha'\right) -U\left( y,\alpha \right) =\underset{
\text{=A: learning effect}}{\underbrace{U\left( y',\alpha'\right)
-U\left( y',\alpha \right) }} + 
\underset{\text{=B: fallback value effect}}{\underbrace{U\left( y',\alpha\right) -U\left( y,\alpha\right) }}.
\end{equation}
Figure \ref{fig: U(y,a) decomposition} illustrates the properties of $U(y,\alpha)$, as well as the decomposition result.   By Lemma \ref{lem: U(y,a)}, the fallback value effect B is always negative, as 
\begin{equation}\label{eq: fallback value effect}
B=U\left( y',\alpha \right) -U\left( y,\alpha \right) =-\sum_{z=y}^{y'-1}\bar{F}_{\alpha}(z) \le 0.
\end{equation}
However, the sign of the learning effect A is uncertain.  Observe that
\begin{equation}\label{eq: learning effect}
A=U\left( y',\alpha'\right) -U\left( y',\alpha \right) =\left( \alpha-\alpha' \right) \Delta(y'),
\end{equation}
and $\Delta(y')\ge 0$ by Lemma \ref{lem: Delta(y)}.  The sign of the learning effect is then determined by $\alpha'-\alpha$.  By \eqref{eq: a_t+1 lower} and \eqref{eq: x_hat}, the posterior shift under lower-censoring is characterized by
\begin{equation}\label{eq: belief shift_lower}
\begin{cases}
\alpha'\ge \alpha,
& \text{if } \{X\le y\} \text{ or } X=x\in \{y+1,\ldots,\hat{x}\},\\[0.2em]
\alpha'<\alpha,
& \text{if } X=x>\max\{y,\hat{x}\}.
\end{cases}
\end{equation}
Since $\alpha'$ is the posterior probability of the inferior distribution $L_1$, an increase (decrease) in $\alpha$ corresponds to a more pessimistic (optimistic) belief.  Equation \eqref{eq: belief shift_lower} shows that the DM can become strictly more optimistic only after a revealed outcome  $x>\max\{y, \hat{x}\}$. 

%%%%%%%%%%%%%%%%%%%%%%%%

By Definition~\ref{def: Monotone Search Problem}, 
if the joint effect of the fallback value and the learning effect is negative for any $\left( y,\alpha \right) $ and any possible sample observation, the search problem is monotone.  The next theorem characterizes monotonicity under lower-censoring in terms of the survival functions (see \eqref{eq:survival}). It also provides a state-independent sufficient condition expressed solely in terms of the model primitives.

\begin{theorem}[Monotonicity under lower-censoring]\label{thm: monotone conditions_lower}
Suppose Assumption \ref{ass: single-crossing} holds. Any search problem \eqref{eq:Bellman-general}, under lower-censoring and Bayesian updating, with any $\a_0$, is monotone if and only if
\begin{equation}\label{eq: IFF_lower}
\sum_{z=y}^{x-1}\bar{F}_{\alpha}(z)
\ge
\frac{\alpha(1-\alpha)}{\pi_{\alpha}(x)}(\pi_{2,x}-\pi_{1,x})\Delta(x)= (\a-\a')\Delta(x).
\end{equation}
for every $y\in\bar{K}$, $\alpha\in[0,1]$ and every $x\in K$ such that $x>\max\{y,\hat{x}\}$. Moreover, a sufficient condition for \eqref{eq: IFF_lower}, depending only on the primitive distributions, is
\begin{equation}\label{eq: SC_lower}
\bar{F}_1(x-1)
\ge
\frac{\sqrt{\pi_{2,x}}-\sqrt{\pi_{1,x}}}{\sqrt{\pi_{2,x}}+\sqrt{\pi_{1,x}}}\,\Delta(x),
\qquad \text{for all } x>\hat{x} .
\end{equation}
\end{theorem}

As noted above, the fallback-value effect, $U(y',\alpha)-U(y,\alpha)$, is always negative.  Under Assumption \ref{ass: single-crossing}, the learning effect can be positive only when the sample value $x>\max\{y,\hat{x}\}$.  Conditions \eqref{eq: IFF_lower} and \eqref{eq: SC_lower} then provide relevant conditions for the negative fallback value effect to dominate the positive learning effect in this region.  Specifically, for any $x>\max\{y,\hat{x}\}$, the LHS of \eqref{eq: IFF_lower} measures the exact fallback value effect, which is equal to the sum of the survival functions from $y$ to $x-1$.  The RHS is the exact learning effect, which is positive as $\pi_{2,x}-\pi_{1,x}>0$.  Intuitively, when the two distributions $L_1$ and $L_2$ are close to each other, the value of learning is small.  In the extreme case of $L_1=L_2$, there is no learning at all and the learning effect is always zero, in which case condition \eqref{eq: IFF_lower} always holds.

The single-crossing order plays two distinct roles. First, it implies FOSD, so $F_1(y)\ge F_2(y)$ for every $y$. Therefore, a censored observation weakly shifts beliefs toward the inferior distribution $L_1$. Second, it implies the increasing-convex order, so $\Delta(y)=U_2(y)-U_1(y)\ge 0$ for every $y$. This ensures that a decrease in $\alpha$ has a nonnegative effect on the marginal value of continued search. These two implications make it possible to identify the potential failure of monotonicity in the high-outcome region $x>\max\{y,\hat{x}\}$.  The single-crossing order also helps to characterize comparative static results of the optimal stopping rule.

Finally, the state-independent sufficient condition~\eqref{eq: SC_lower} depends only on the model primitives, such as the distributions $L_1$ and $L_2$.
This condition is state-free because it bounds the two sides of \eqref{eq: IFF_lower} uniformly over $(y,\alpha)$.  To be specific, the LHS term $\bar{F}_1(x-1)$ is a lower bound on the magnitude of the negative fallback-value effect, while the RHS is an upper bound on the positive learning effect.   \ignore{Note that, under Assumption \ref{ass: single-crossing} and that $x>\hat{x}$, we have $\pi_{2,x}>\pi_{1,x}$, and hence
\begin{equation*}
    \frac{\sqrt{\pi_{2,x}}-\sqrt{\pi_{1,x}}}{\sqrt{\pi_{2,x}}+\sqrt{\pi_{1,x}}}=\frac{\pi_{2,x}-\pi_{1,x}}{(\sqrt{\pi_{2,x}}+\sqrt{\pi_{1,x}})^2}<\frac{\pi_{2,x}-\pi_{1,x}}{(2\sqrt{\pi_{1,x}})^2}=\frac{\pi_{2,x}-\pi_{1,x}}{4\pi_{1,x}}.
\end{equation*}
A stronger sufficient condition that implies \eqref{eq: SC_lower} is as follows: under Assumption \ref{ass: single-crossing},
\begin{equation}\label{eq: SC 2}
1-F_{1}\left( x-1\right) \geq \frac{\pi_{2,x}-\pi_{1,x}}{4\pi_{1,x}}\Delta \left( x\right),\quad\text{for any $x>\hat{x}$}.
\end{equation}
Condition \eqref{eq: SC_lower} is easy to verify because it depends only on model primitives.  
} 
The following example shows how to use the sufficient condition \eqref{eq: SC_lower} to verify a monotone search problem.

\begin{example}[Monotone under lower-censoring]\label{ex: monotone search}
Consider the following example: 
\begin{equation*}
    L_1=(0.25, 0.25, 0.25, 0.25), \qquad L_2=(0.2, 0.1, 0.3, 0.4).
\end{equation*}
\ignore{
\begin{equation*}
\begin{tabular}{|c|c|c|c|c|}
\hline
$X$ & $1$ & $2$ & $3$ & $4$ \\ \hline
$L_{1}$ & $\pi_{1,1}=0.25$ & $\pi_{1,2}=0.25$ & $\pi_{1,3}=0.25$ & $\pi_{1,4}=0.25$ \\ \hline
$L_{2}$ & $\pi_{2,1}=0.2$ & $\pi_{2,2}=0.1$ & $\pi_{2,3}=0.3$ & $\pi_{2,4}=0.4$ \\ \hline
\end{tabular}
\end{equation*}
}
Observe that $L_{1}\prec _{sc}L_{2}$ and $\hat{x}=2$.  When $x=4$, $\Delta \left( x\right) =0$ and the sufficient condition (\ref{eq: SC_lower}) holds obviously.  When $x=3$, $1-F_{1}\left( x-1\right) =0.5$ and 
\begin{equation*}
\frac{\sqrt{\pi_{2,x}}-\sqrt{\pi_{1,x}}}{\sqrt{\pi_{2,x}}+
\sqrt{\pi_{1,x}}}\Delta \left( x\right) =\frac{\sqrt{0.3}-\sqrt{0.25}}{\sqrt{%
0.3}+\sqrt{0.25}}\left( 0.4-0.25\right) \left( 4-3\right) \simeq 0.0068.
\end{equation*}
Condition (\ref{eq: SC_lower}) also holds.  Therefore, the search problem under lower-censoring is monotone.
\end{example}

The single-crossing order (Assumption \ref{ass: single-crossing}) can not be weakened to FOSD.  The main reason is that, under FOSD, the sequence $(\pi_{1,i}-\pi_{2,i})_{i=1}^k$ can change its signs multiple times, and hence, induce extra volatility of belief shifts in the region of $x\in\{y+1,\ldots,\hat{x}\}$, which may destroy the monotonicity that holds under the single-crossing assumption. Consider the following example,
\begin{equation*}
    L_1=(0.9, 0.04, 0.05, 0.01), \qquad L_2=(0.1, 0.6, 0.04, 0.26).
\end{equation*}
Observe that $L_{1}\prec _{st}L_{2}$, but the single-crossing assumption is not satisfied.  The sufficient condition \eqref{eq: SC_lower} satisfies obviously as $\hat{x}=3$ by definition \eqref{eq: x_hat} and $\Delta(4)=0$, yet the search problem is not monotone.  To see this, consider an initial state $(y,\alpha)=(1,\frac{9}{10})$ and a sample value $X=2>y$.  The updated state is $(y',\alpha')=(2,\frac{6}{16})$, and the change in one-step marginal search value is:
\begin{equation*}
    U(y',\alpha')-U(y,\alpha)=A+B=(\alpha-\alpha')\Delta(y')-\bar{F}_{\alpha}(y)=\frac{21}{40}\times \frac{49}{100}-\frac{18}{100}>0.
\end{equation*}
Therefore, this search problem is not monotone under lower-censoring.

\paragraph{Comparison with Bikhchandani and Sharma (1996).} \cite{bikhchandani1996optimal} (BS) study a sequential search model with an \emph{ad hoc}, rather than Bayesian, learning rule. They propose two assumptions on the learning rule that are sufficient for monotonicity of the search problem.

BS consider an information structure with full revelation, in which all search outcomes are observed by the DM. Let the search history at period $t$ be $h_t=\mathbf{x}_t=(x_1,\ldots,x_t)$, and let the corresponding fallback value be $y=\max\{x_1,\ldots,x_t\}<k$. Removing the most recent observation $x_t$ yields the period-$t-1$ history $\mathbf{x}_{t-1}=(x_1,\ldots,x_{t-1})$. Assumption~1 of BS requires that, for every outcome value $y'>y$, the posterior probability
\begin{equation*}
    \Pr(X=y'|\mathbf{x}_t)\leq \Pr(X=y'|\mathbf{x}_{t-1}).
\end{equation*}
It formalizes the notion that if a DM observes more outcomes that are lower than $y$, then the posterior probability of observing a value higher than $y$ decreases.  In the language of our decomposition, Assumption 1 of BS rules out a positive learning effect after an observation below the current fallback value.

Moreover, consider two realized search histories of the same length, e.g., $\mathbf{x}_t=(x_1,x_2,\ldots,x_t)$ and $\tilde{\mathbf{x}}_t=(\tilde{x}_1,\tilde{x}_2,\ldots,\tilde{x}_t)$, and let $y<k$ be the maximum observation of the two search histories.  Assumption 2 of BS states that, for any $y'>y$,
\begin{equation*}
    \Pr(X=y'|\mathbf{x}_t)= \Pr(X=y'|\tilde{\mathbf{x}}_{t}).
\end{equation*}
That is, the posterior probability of observing a value higher than $y$, given that all the previous observations are lower than $y$, depends only on the number of previous observations.  This assumption is pivotal to their result of a history-independent cutoff for optimal stopping.

Neither Assumption 1 nor Assumption 2 of BS holds in our model.  We consider a Bayesian learning model under lower-censoring, which is distinct from the \emph{ad hoc} learning rule.  Moreover, given the censoring rule in our model, the sequence of observations matters in determining the posterior beliefs.  We next show this using the previous Example \ref{ex: monotone search}.

Let $\a_t$ denote the posterior belief after observing the search history $\mathbf{x}_t$, and the prior belief $\alpha=0.5$.  Consider two search histories: $\mathbf{x}_2=(x_1,x_2)=(1,3)$ and $\tilde{\mathbf{x}}_2=(\tilde{x}_1,\tilde{x}_2)=(2,1)$.  Along the search history $\mathbf{x}_2=(1,3)$, from \eqref{eq: a_t+1 lower}, the posterior beliefs are
\begin{equation*}
  \alpha_1=\frac{\alpha \pi_{1,1}}{\alpha \pi_{1,1}+(1-\alpha) \pi_{2,1}}=\frac{5}{9},\quad \alpha_2 = \frac{\alpha_1 \pi_{1,3}}{\alpha_1 \pi_{1,3}+(1-\alpha_1) \pi_{2,3}}=\frac{25}{49}.
\end{equation*}
The predictive probabilities of observing $X=4$ are
\[
\pi_{\a_1}(4)=\a_1\times \pi_{1,4}+(1-\a_1)\times \pi_{2,4} =\frac{57}{180},\quad\pi_{\a_2}(4)=\a_2\times \pi_{1,4}+(1-\a_2)\times \pi_{2,4} =\frac{317}{980}.
\]
Assumption 1 of BS does not hold as $\pi_{\a_1}(4)<\pi_{\a_2}(4)$, i.e., it is more likely to observe a higher value $4$ with more observations of values strictly less than $4$.  

Along the search history $\tilde{\mathbf{x}}_2=(2,1)$, from \eqref{eq: a_t+1 lower}, the posterior beliefs are
\begin{equation*}
  \tilde{\alpha}_1=\frac{\alpha \pi_{1,2}}{\alpha \pi_{1,2}+(1-\alpha) \pi_{2,2}}=\frac{5}{7},\quad \tilde{\alpha}_2 = \frac{\tilde{\alpha}_1 F_1(2)}{\tilde{\alpha}_1 F_1(2)+(1-\tilde{\alpha}_1) F_2(2)}=\frac{25}{31}.
\end{equation*}
Assumption 2 of BS does not hold either as $\alpha_2\neq \tilde{\alpha}_2$, which implies $\pi_{\a_2}(4)\neq \pi_{\tilde{\alpha}_2}(4)$.  

This comparison highlights the distinction between our Bayesian learning model under lower-censoring and the \emph{ad hoc} learning rule considered by BS. In our model, the learning effect can be positive, and posterior beliefs depend on the order of observations. Consequently, the sufficient condition \eqref{eq: SC_lower} accommodates cases that are excluded by Assumption 1 of BS, and the resulting optimal stopping cutoffs are generally history-dependent, unlike the history-independent cutoffs obtained under Assumption 2 of BS.

\subsection{On the optimal stopping rule}
Proposition~\ref{prop: optimal stopping monotone} characterizes the optimal stopping time in monotone search problems. The optimal rule is myopic and takes the form of a cutoff rule. For any given $\alpha\in[0,1]$, define the cutoff value
\begin{equation}\label{eq: optimal stopping cutoff}
y^{\ast}(\alpha)=\min\left\{z\in\bar{K}:U(z,\alpha)\leq c\right\}.
\end{equation}
The uniqueness of $y^*(\alpha)$ follows from the fact that the function $z\mapsto U(z,\alpha)$ is strictly decreasing for all $z<k$, while $U(k,\alpha)=0$.  By Proposition~\ref{prop: optimal stopping monotone}, it is optimal to stop if and only if $y\geq y^*(\alpha)$. 

For comparison, it is also useful to define the corresponding cutoff values under each distribution $L_i$:
\[
y_i^*=\min\left\{z\in\bar{K}:U_i(z)\leq c\right\},\qquad i=1,2.
\]
The value $y_i^*$ can be interpreted as the optimal stopping cutoff when the underlying distribution is known to be $L_i$.  These benchmark cutoffs serve as reference points for comparing the optimal stopping rule under uncertainty about $L$ with the corresponding stopping rules in the full-information environments where the distribution is known.

\begin{lemma}[Stopping cutoff]\label{lem: y*(a)}
Suppose Assumption \ref{ass: single-crossing} holds and the search problem is monotone.  (i) Then for every $\alpha\in[0,1]$, the stopping cutoff
\begin{equation}\label{eq: y*(a)}
    y_1^*\le y^*(\alpha)\le y_2^*.
\end{equation}
(ii) $y^*(\alpha)$ is weakly decreasing in both $\alpha$ and $c$. For a fixed $c$ it is a step function of $\alpha$.  At the boundary beliefs, $y^*(0)=y_2^*$ and $y^*(1)=y_1^*$.\\
(iii) Let $\alpha'$ denote the posterior belief that follows $\alpha$. Then,
\begin{equation*}
\begin{cases}
    y^*(\alpha')\le y^*(\alpha),
        & \text{if } \{X\le y\} \text{ or } X=x\in \{y+1,\ldots,\hat{x}\},\\[0.4em]
    y^*(\alpha')\ge y^*(\alpha),
        &  \text{if } X=x>\max\{y,\hat{x}\}.
\end{cases}
\end{equation*}
\end{lemma}

Consider result (iii).  The first case corresponds to a posterior shift toward the inferior distribution $L_1$, which lowers the one-step marginal search value and therefore weakly lowers the stopping cutoff. The second case corresponds to a posterior shift toward the superior distribution $L_2$, which raises the one-step marginal search value and therefore weakly raises the stopping cutoff. Equality may occur because $y^*(\cdot)$ is a discrete step function, or because the posterior belief remains unchanged, e.g., when $\pi_{1,x}=\pi_{2,x}$ for a fully revealed observation.

It is also useful to represent the stopping rule in terms of belief cutoffs.  By \eqref{eq: y*(a)}, the DM continues searching for every belief whenever $y<y_1^*$.  We define the set of fallback values for which stopping may occur by $\mathcal Y:=\{y\in \bar{K}:y\ge y_1^*\}$.  For each $y\in\mathcal Y$, define
\begin{equation}\label{eq: belief cutoff}
 \underline\alpha(y):=\inf\{\alpha\in[0,1]: U(y,\alpha)\le c\}. 
\end{equation}
If $U(y,0)\le c$, then $\underline\alpha(y)=0$, because stopping is optimal for every belief. If 
$U(y,1)=c<U(y,0)$, then $\underline\alpha(y)=1$, because stopping is optimal only at the degenerate belief 
$\alpha=1$.  When $U(y,1)<c<U(y,0)$, then the belief cutoff is interior and is given as follows:
\begin{equation}\label{eq:belief-cutoff-explicit}
    \underline\alpha(y)
    =\frac{U_2(y)-c}{U_2(y)-U_1(y)}.
\end{equation}

%If $U(y,0)\le c$, then $\underline\alpha(y)=0$; if $U(y,1)>c$, then no belief in $[0,1]$ induces stopping at $y$, so $y\notin\mathcal Y$ under the maintained definition.

\begin{lemma}[Belief cutoff for optimal stopping]\label{lem:belief-cutoff-stopping}
Suppose Assumption \ref{ass: single-crossing} holds and the search problem is monotone.  For any $y\in\mathcal Y$, it is optimal to stop if and only if $\alpha\ge\underline\alpha(y)$.
\end{lemma}

Figure \ref{fig: stopping region} illustrates the stopping region of a monotone search problem in the state space of $(y,\alpha)$.  One can check that this example satisfies the sufficient condition \eqref{eq: SC_lower} for monotonicity.

\begin{figure}[t]
\begin{tikzpicture}
\begin{axis}[
    width=12cm,
    height=8cm,
    xlabel={Posterior belief $\alpha$},
    ylabel={Fallback value $y$},
    xmin=0,
    xmax=1.0,
    ymin=0,
    ymax=10.5,
    xtick={0,0.1,0.2,0.3,0.4,0.5,0.6,0.7,0.8,0.9,1},
    ytick={0,1,2,3,4,5,6,7,8,9,10},
    ymajorgrids=true,
    xmajorgrids=true,
    grid style={dashed,gray!30},
    %legend style={draw=none, fill=none, inner sep=0pt, at={(0.99, 0.16)},anchor=south east},
    %legend cell align={left},
    %title={Stopping region for $c=4$},
]

% Shaded stopping region: $U(y,\alpha)\le 4$
\addplot[draw=none, fill=blue!20] coordinates {(0.868687,0.5) (1.000000,0.5) (1.000000,1.5) (0.868687,1.5)} -- cycle;
\addplot[draw=none, fill=blue!20] coordinates {(0.495726,1.5) (1.000000,1.5) (1.000000,2.5) (0.495726,2.5)} -- cycle;
\addplot[draw=none, fill=blue!20] coordinates {(0.095238,2.5) (1.000000,2.5) (1.000000,3.5) (0.095238,3.5)} -- cycle;
\addplot[draw=none, fill=blue!20] coordinates {(0.000000,3.5) (1.000000,3.5) (1.000000,4.5) (0.000000,4.5)} -- cycle;
\addplot[draw=none, fill=blue!20] coordinates {(0.000000,4.5) (1.000000,4.5) (1.000000,5.5) (0.000000,5.5)} -- cycle;
\addplot[draw=none, fill=blue!20] coordinates {(0.000000,5.5) (1.000000,5.5) (1.000000,6.5) (0.000000,6.5)} -- cycle;
\addplot[draw=none, fill=blue!20] coordinates {(0.000000,6.5) (1.000000,6.5) (1.000000,7.5) (0.000000,7.5)} -- cycle;
\addplot[draw=none, fill=blue!20] coordinates {(0.000000,7.5) (1.000000,7.5) (1.000000,8.5) (0.000000,8.5)} -- cycle;
\addplot[draw=none, fill=blue!20] coordinates {(0.000000,8.5) (1.000000,8.5) (1.000000,9.5) (0.000000,9.5)} -- cycle;
\addplot[draw=none, fill=blue!20] coordinates {(0.000000,9.5) (1.000000,9.5) (1.000000,10.5) (0.000000,10.5)} -- cycle;

% Blue dashed boundary of the discrete stopping region
\addplot[blue!25, dashed, very thick] coordinates {(1,0.5) (0.868687,0.5) (0.868687,1.5) (0.495726,1.5) (0.495726,2.5) (0.095238,2.5) (0.095238,3.5) (0.000000,3.5) (0.000000,10.5)};

% Belief cutoff boundary
\addplot+[only marks, mark=*, mark size=2.2pt] coordinates {(0.868687,1) (0.495726,2) (0.095238,3) (0.000000,4) (0.000000,5) (0.000000,6) (0.000000,7) (0.000000,8) (0.000000,9) (0.000000,10)};

% Labels for nontrivial cutoff points
\node[anchor=west, font=\small] at (axis cs:0.87,1) {$\underline\alpha(1)$};
\node[anchor=west, font=\small] at (axis cs:0.50,2) {$\underline\alpha(2)$};
\node[anchor=west, font=\small] at (axis cs:0.10,3) {$\underline\alpha(3)$};

\node[anchor=west] at (axis cs:0.4,8) {Stopping region};
%\node[anchor=west, font=\scriptsize] at (axis cs:1.01,0) {no stop};

\end{axis}
\end{tikzpicture}
\caption{Stopping region: $\{(y,\alpha): U(y,\alpha)\le c\}$}\label{fig: stopping region}
\floatfoot{In this example, $c=4$, and the distributions $L_1$ and $L_2$ are the same as in Figure \ref{fig: U(y,a) decomposition}.  
}
\end{figure}
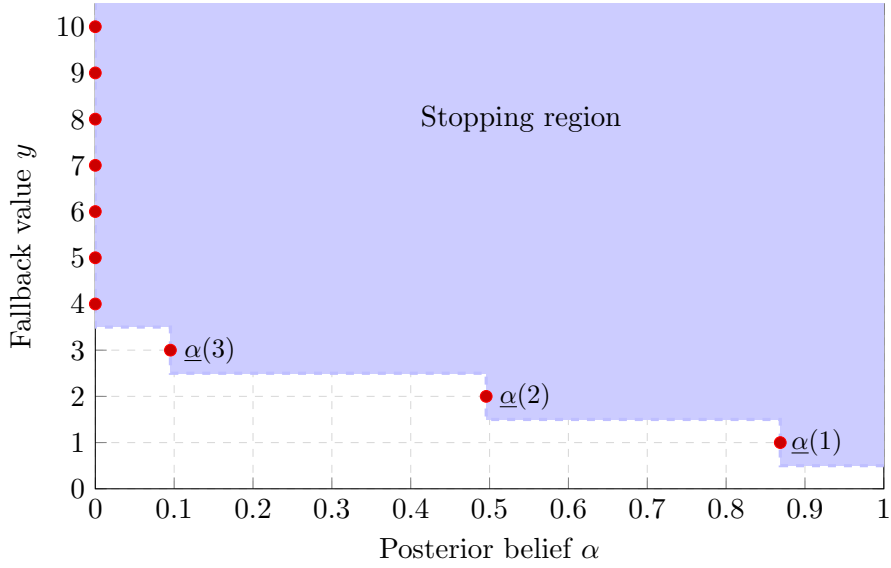

To study the process of posterior updates, it is convenient to use the likelihood ratios.  For $\alpha\in(0,1)$, define the following likelihood ratio $    O(\alpha):=\frac{\alpha}{1-\alpha}$.   
\ignore{
Under lower-censoring, the odds ratio evolves according to
\begin{equation}\label{eq:odds-ratio-update}
    \frac{O(\alpha')}{O(\alpha)}
    =
    \begin{cases}
        \dfrac{F_1(y)}{F_2(y)}, & \text{after the censored event }\{X\le y\},\\[0.5em]
        \dfrac{\pi_{1,x}}{\pi_{2,x}}, & \text{after the fully revealed observation }x>y.
    \end{cases}
\end{equation}
}
For any belief cutoff $\underline\alpha(y)\in(0,1)$, define the corresponding likelihood cutoff
\begin{equation}
    \bar O(y):=\frac{\underline\alpha(y)}{1-\underline\alpha(y)}.
    \label{eq:odds-cutoff}
\end{equation} 
The optimal stopping rule can be equivalently defined by stopping at $(y,\alpha)$ if and only if $O(\alpha)\ge \bar O(y).$ Note that $\bar O(y)$ is strictly decreasing in $c$, as $\underline{\alpha}(y)$ is strictly decreasing in $c$. 

Define $R(y):=\frac{F_1(y)}{F_2(y)}$ for all $y\in K$.  Under Assumption \ref{ass: single-crossing}, $F_1(y)\ge F_2(y)$, so $R(y)\ge1$. Moreover, since $\pi_{1,k}<\pi_{2,k}$, single crossing implies $F_1(y)>F_2(y)$ for every $y<k$; hence $R(y)>1$ for all $y<k$.  If the fallback value remains $y$ and the DM observes $m$ consecutive censored events, then the likelihood ratio becomes $O(\alpha)R(y)^m$.  For $y\in\mathcal Y$ and $\alpha\in(0,1)$, we can define the minimum number of consecutive censored observations that induces stopping by
\begin{equation}\label{eq: consecutive censoring bound_lower}
    m_y(\alpha)
    :=\min\Bigl\{m\in\mathbb N_0:
    O(\alpha)R(y)^m\ge \bar O(y)\Bigr\}.
\end{equation}
For instance, if $\bar O(y)=0$, then $m_y(\alpha)=0$. Likewise, if $\bar O(y)=\infty$, then $m_y(\alpha)=\infty$. Observe that $m_y(\alpha)$ is decreasing in $\a$ and $c$, because $O(\alpha)$ is increasing in $\a$ and $\bar O(y)$ is strictly decreasing in $c$.

The case $\underline\alpha(y_1^*)=1$ is a knife-edge case in which, under $L_1$, the DM is indifferent between stopping and continuing search at the cutoff value $y_1^*$. The following genericity assumption rules out this case by ensuring that stopping is strictly preferred at $y_1^*$.

\begin{assumption}\label{ass:strict-lowest-cutoff}
    $U_1(y_1^*)<c$, or equivalently $\underline\alpha(y_1^*)<1$.
\end{assumption}

Since $U_1(y)$ is decreasing in $y$, Assumption \ref{ass:strict-lowest-cutoff} implies $U_1(y)<c$ for every
$y\ge y_1^*$. Hence $\underline\alpha(y)<1$ and $\bar O(y)<\infty$ for all $y\in\mathcal Y$.  Starting from any state $y\in\mathcal Y$ and $\alpha\in(0,1)$, a sequence of $m_y(\alpha)$ consecutive censored observations will induce stopping. However, an uncensored observation will raise the fallback value.  As the set $K$ is finite, a finite number of uncensored observations will also induce stopping.  As a result, the optimal policy must stop after a finite number of additional searches from any state $y\in\mathcal Y$ and $\alpha\in(0,1)$.

\begin{proposition}[Finite pathwise bound]\label{prop:finite-pathwise-bound}
Suppose Assumption \ref{ass: single-crossing}, Assumption \ref{ass:strict-lowest-cutoff}, and condition \eqref{eq: SC_lower} hold. Starting from any state $(y,\alpha)$ with $y\in\mathcal Y$ and $\alpha\in(0,1)$, the optimal policy stops after at most $\bar n(y,\alpha)$ additional searches on every feasible search history, where $\bar n(y,\alpha)<\infty$.

Moreover, $\bar n(y,\alpha)=0$ for all $y\ge y_2^*$. For $y\in\mathcal Y$ with $y<y_2^*$, define
\begin{equation}\label{eq: alpha(m,x)}
    \alpha^{(m,x)}
    :=
    \frac{
        O(\alpha)R(y)^m(\pi_{1,x}/\pi_{2,x})
    }{
        1+O(\alpha)R(y)^m(\pi_{1,x}/\pi_{2,x})
    },
    \qquad
    m=0,\ldots,m_y(\alpha)-1,\quad x>y.
\end{equation}
Then one valid recursive bound is
\begin{equation}\label{eq: pathwise bound}
\bar n(y,\alpha)
:=
\max\left\{
    m_y(\alpha),
    \max_{0\le m<m_y(\alpha)}
    \max_{x\in\{y+1,\ldots,k\}}
    \left[
        m+1+\bar n\bigl(x,\alpha^{(m,x)}\bigr)
    \right]
\right\}.    
\end{equation}
Moreover, $\bar n(y,\alpha)$ is decreasing in both $\alpha$ and in $c$.
\end{proposition}

The restriction $\alpha\in(0,1)$ excludes degenerate priors of $\alpha=0$ or $1$, for which no intertemporal learning occurs.  Assumption \ref{ass:strict-lowest-cutoff} rules out the knife-edge case $\underline\alpha(y_1^*)=1$, in which the bound $m_{y_1^*}(\alpha)$ can be infinite.  With free recall, the fallback value weakly increases over time.  Under Assumption \ref{ass: single-crossing}, censored events are relatively more likely under the inferior distribution $L_1$, so they multiply the odds ratio by $R(y)>1$ and push beliefs toward $L_1$.  By contrast, sufficiently high uncensored observations can push beliefs toward $L_2$, but they also raise the fallback value, which moves the state closer to the stopping region.

\subsection{Extension to Continuous Distributions}\label{subsec: continuous version}

The baseline model of discrete distributions can be naturally extended to continuous distributions.  Now consider two continuous distributions $F_1$ and $F_2$ defined on a real interval $[0,b]$, with strictly positive densities $f_1$ and $f_2$ and $f_1(b)<f_2(b)$.  Given a posterior belief $\alpha$, we define $F_{\alpha}$, $\bar{F}_{\alpha}$ and $f_{\alpha}$ in the same way as before.  Given a state $(y,\alpha)$, the one-step marginal search value is
\[
U(y,\alpha)=\E_{\a}[(X-y)_+]=\int_y^b\bar{F}_{\a}(z)dz, \text{ and } \Delta(y)=U_2(y)-U_1(y)=\int_y^b(F_1(z)-F_{2}(z))dz.
\]
We also assume $F_1$ is smaller than $F_2$ in single-crossing order, i.e., $f_1(x)-f_2(x)$ changes its sign from nonnegative to nonpositive only once on $[a,b]$, and define $\hat{x}:=\max\{x\in[0,b]: f_1(x)- f_2(x)\ge 0\}$.  The updated posterior under lower-censoring is
\begin{equation*}
    \alpha'=
    \begin{cases}
        \dfrac{\alpha F_1(y)}{F_{\alpha}(y)},
        & \text{if }\{X\le y\}\text{ is observed},\\[0.5em]
        \dfrac{\alpha f_{1}(x)}{f_{\alpha}(x)},
        & \text{if } X=x>y\text{ is observed}.
    \end{cases}
\end{equation*}
Similar to Theorem \ref{thm: monotone conditions_lower}, we have the following monotone conditions for continuous distributions. 

\setcounter{theorem}{0}          % Next \refstepcounter will yield 1
\begingroup                      % Start local group
  \renewcommand{\thetheorem}{\arabic{theorem}*} % Append * to the number
  \begin{theorem}[Monotonicity under lower-censoring: continuous distributions]\label{thm: monotone conditions_lower_continuous}
Suppose Assumption \ref{ass: single-crossing} holds. For continuous distributions, any search problem \eqref{eq:Bellman-general} under lower-censoring and Bayesian updating, with any $\alpha_0$, is monotone if and only if
\begin{equation}\label{eq: IFF_lower_continuous}
        \int_y^x \bar{F}_\alpha(z) \, dz \ge \frac{\alpha(1-\alpha)}{f_\alpha(x)} \bigl( f_2(x) - f_1(x) \bigr) \Delta(x).
\end{equation}
for all $y \in [0,b]$, $\alpha \in [0,1]$, and every $x > \max\{y, \hat{x}\}$. 
\end{theorem}
\endgroup

%%%%%%%%%%%%%%%%%%%%%%%%%%%%%%%%

\section{An Application: Job Search}\label{sec: job search}  

Our model of Bayesian search under lower-censoring readily applies to job search environments, where the recruiting process naturally generates an endogenous lower-censoring structure tied to the job seeker's evolving reservation utility. Consider a worker whose current best option, either her current employment contract or the highest firm offer received to date, has value $y \in \overline{K}$. When evaluating an applicant, a recruiting firm first requests the worker's current baseline requirement $y$ (e.g., her current compensation or minimum salary expectation) to determine whether it can offer a competitive match.
The application then proceeds in two stages:
\paragraph{Screening against the dynamic baseline ($X_i \le y$):} If the preliminary evaluation reveals that the match value fails to clear the applicant's current fallback threshold ($X_i \le y$), the firm terminates the process without investing in detailed evaluation or wage design. The applicant receives only a standard rejection, learning the censored event $\{X_i \le y\}$. She does not observe the exact match realization $X_i$ because the firm has no economic incentive to precisely assess or communicate an offer that it knows the worker will reject.
\paragraph{Detailed evaluation and negotiation ($X_i > y$):} If the match value exceeds the worker's fallback option ($X_i > y$), the firm advances to intensive interviews and formal offer negotiations. In this stage, detailed information regarding match quality, performance bonuses, and job fit is uncovered, allowing the worker to observe the exact realization $X_i = x > y$.

Because the worker's fallback value $y$ updates dynamically over time whenever a superior offer is discovered ($y' = \max\{y, X_i\}$), the threshold determining which outcomes are censored and which are fully revealed moves endogenously with the search history. Thus, the job seeker obtains precise match information only when a firm's prospective offer strictly exceeds her prevailing fallback requirement.

Suppose the monotonicity conditions developed above hold. If the job seeker's current state is $(y,\alpha)$, then her optimal search strategy is characterized by the stopping cutoff $y^*(\alpha)$. She continues searching if $y<y^*(\alpha)$ and stops if $y\ge y^*(\alpha)$. Equivalently, for a fixed fallback value $y$, she stops whenever her posterior belief in the inferior distribution exceeds the belief cutoff $\underline\alpha(y)$. Consecutive rejections raise the posterior odds of $L_1$ by the factor
\[
    R(y):=\frac{F_1(y)}{F_2(y)}>1 \text{ for } y<k.
\]
Hence a sequence of rejections eventually makes the job seeker sufficiently pessimistic about the outside market that she stops searching and remains in her current job.

\begin{example}[Job search]\label{ex:job-search}
Let $K=\{1,2,3,4\}$, $c=1$, $y=2$, and $\alpha=0.05$. Suppose
\[
    L_1=(0.4,0.4,0.1,0.1),
    \qquad
    L_2=(0.1,0.1,0.4,0.4).
\]
Then $\hat{x}=2$, and one can verify that the state-free sufficient condition \eqref{eq: SC_lower} is satisfied. The search problem is therefore monotone under lower-censoring.

We have $U_1(2)=0.3$ and $U_2(2)=1.2$.  Therefore, $U(2,0.05)=0.05\cdot 0.3+0.95\cdot 1.2=1.155>c$, and $U(3,0.05)=0.05\cdot 0.1+0.95\cdot 0.4=0.385<c$. Observing that the stopping cutoff is $y^*(0.05)=3$, which is greater than the fallback value $y=2$, the job seeker will continue searching.  For a given fallback value $y=2$, by \eqref{eq:belief-cutoff-explicit}, the belief cutoff for optimal stopping is
\[
    \underline\alpha(2)
    =\frac{U_2(2)-c}{U_2(2)-U_1(2)}
    =\frac{1.2-1}{1.2-0.3}
    =\frac{2}{9}.
\]
The corresponding odds cutoff is $\bar O(2)=\frac{\underline\alpha(2)}{1-\underline\alpha(2)}=\frac{2}{7}$.  The initial odds ratio is $O(0.05)=\frac{0.05}{0.95}=\frac{1}{19}$,  and a rejection at fallback value $y=2$ multiplies the odds ratio by 
\[
    R(2)=\frac{F_1(2)}{F_2(2)}=\frac{0.8}{0.2}=4.
\]
After one rejection, the odds ratio is $\frac{4}{19}<\frac{2}{7}$, so the job seeker continues searching. After two consecutive rejections, the odds ratio is $\frac{16}{19}>\frac{2}{7}$, so the posterior belief crosses the stopping cutoff. Therefore, if the job seeker receives two consecutive rejections while her fallback value remains $y=2$, she stops searching and stays in her current job.
\end{example}

%%%%%%%%%%%%%%%%%%%%%%%%%%%%%%%%%%

\section{Upper Censorship: Consumer Price Search}\label{sec: consumer search}
While the preceding sections focused on lower censoring, we now examine its theoretical dual: upper censoring, an information structure that arises naturally in consumer price search

A consumer searches sequentially for a lower price among a large number of \emph{ex ante} homogeneous retailers. 
A retailer's price $X\in K$ is an i.i.d. draw from an unknown distribution $L\in\{L_1, L_2\}$, and each price quote costs $c>0$.  The search is with free recall: after any search history, the consumer can stop and purchase at the lowest price observed so far.  Here we consider the information structure of \emph{upper-censoring}. To be specific, let $x$ be the realization of $X$ and $p$ be the lowest price discovered so far.  Under upper-censoring, the signal observed is
\[
S^u(x,p)=
\begin{cases}
x, & \text{if } x<p,\\
\{X\ge p\}, & \text{if } x\ge p,
\end{cases}
\]
where the superscript $u$ means upper-censoring.  After a search history $h_t^u=(S_1^u,\ldots,S_t^u)$, the DM's posterior belief is $\alpha_t:=\Pr(L=L_1\mid h_t^u)$.  For any given state $(p,\alpha_t)$ and a new price quote $X$, the updated posterior belief is given by\footnote{In the hope of not causing confusion, we still use the notation $\a_{t}$, rather than $\alpha^u_t$, in the case of upper-censoring.}
\begin{equation}\label{eq:alpha-update-upper}
\alpha_{t+1}=
\begin{cases}
\displaystyle
\frac{\alpha_t \pi_{1,x}}{\pi_{\a_t}(x)},
& \text{if } X=x<p \text{ is observed}, \\[1.2em]
\displaystyle
\frac{\alpha_t\bar{F}_1(p-1)}{\bar{F}_{\alpha_t}(p-1)},
& \text{if }\{X\ge p\}\text{ is observed}.
\end{cases}
\end{equation}
Define the consumer's fallback price $P_t$ as the lowest price discovered after a search history $h^u_t$.  Under upper-censoring, given $P_0=k+1$, the fallback price $P_t$ then evolves according to
\begin{equation}\label{eq: P_t+1}
P_{t+1}=\min\{P_t,X_{t+1}\}.    
\end{equation}
Let $K^u:=K\cup\{k+1\}$ denote the set of possible fallback prices.

\ignore{We can formulate the MDP problem in a similar way.  Given any state $(p,\alpha)$ and a next price quote $X$, the signal observed is $S^u$.  The updated belief is $\alpha'=B(y,\alpha,S^u)$ as given by \eqref{eq:alpha-update-upper}, and the updated fallback price $p'=\min\{p, X\}$ is given by \eqref{eq: P_t+1}.  Let  $J^u(p,\alpha)$ denote the minimum expected price, and the Bellman equation can be written as
\begin{equation}\label{eq:Bellman-upper}
   J^u(p,\alpha)
    =
    \max\left\{
        p,
        c+
        \E_{\boldsymbol\alpha}\left[
            J^u\bigl(\min\{p,X\},B(p,\alpha,S^u)\bigr)
        \right]
    \right\}.
\end{equation}
}
For consumer price search problems, we define the one-step marginal value of search as the expected reduction in the fallback price from one additional price quote, that is, for $p\in K^u$,
\begin{equation}\label{eq:U-upper}
    U^u(p,\alpha)
    :=\E_{\alpha}[(p-X)_+]
    =\sum_{z=1}^{p-1}\pi_{\alpha}(z)(p-z)
    =\alpha U^u_1(p)+(1-\alpha)U^u_2(p),
\end{equation}
where $U^u_i(p):=\E_{L_i}[(p-X)_+]$.  Similarly, define the difference between $U^u_2(p)$ and $U^u_1(p)$ by
\begin{equation}\label{eq:Delta-upper}
    \Delta^u(p):=U^u_2(p)-U^u_1(p)
    =\sum_{z=1}^{p-1}(\pi_{2,z}-\pi_{1,z})(p-z).
\end{equation}
The following Lemma, analogous to  Lemmata~\ref{lem: Delta(y)} and~\ref{lem: U(y,a)}, lists some properties of $U^u(p,\alpha)$ and $\Delta^u(p)$.

\begin{lemma}\label{lem: Delta(p) U(p,a)_upper}
(1) Under Assumption \ref{ass: single-crossing}, $\bar{F}_{1}(p)\leq\bar{F}_{2}\left(
p\right)$ and $\Delta^u \left( p\right) \leq 0$, and $p\mapsto \Delta^u(p)$ is weakly decreasing.  (2) $U^u\left( p,\alpha \right) $ is increasing and convex in $p$. Moreover, $U^u\left( p,\alpha \right) $ is increasing in $\alpha $, and has increasing differences in $\left( p,\alpha \right) $.
\end{lemma}

Since lower prices are preferred by the consumer, $L_1\prec_{sc}L_2$ implies that $L_1$ is more favorable than $L_2$. Thus, in contrast to the previous payoff-maximizing model, a higher posterior belief $\alpha$ corresponds to a more optimistic posterior predictive distribution from the consumer's perspective.

The price-search problem under upper-censoring is monotone if, for every state $(p,\alpha)$ and every possible next price quote $X$, the induced next state $(p',\alpha')$ satisfies $U^u(p,\alpha)\ge U^u(p',\alpha')$, where $p'$ is derived by \eqref{eq: P_t+1} and $\alpha'$ is the updated belief by \eqref{eq:alpha-update-upper}.  Similarly, we define
\[
\hat{x}^u:=\min\{x\in K: \pi_{1,x}-\pi_{2,x}\le 0\}.
\]
Under Assumption \ref{ass: single-crossing}, any fully revealed value in the region of $x<\hat{x}^u$ will shift the posterior belief toward $L_1$, and therefore increase the one-step marginal search value.  The following proposition gives the monotonicity condition for a price search problem under upper-censoring.

\begin{proposition}\label{prop: mono-condition-upper}
Suppose Assumption \ref{ass: single-crossing} holds. The consumer price-search problem under upper-censoring and Bayesian updating, with any $\alpha_0$, is monotone if and only if
\begin{equation}\label{eq: IFF_upper}
    \sum_{z=x}^{p-1}F_{\alpha}(z)
    \ge
    (\alpha-\alpha')\Delta^u(x)
\end{equation}
for every $p\in K^u$ and $\alpha\in[0,1]$ and every $x\in K$ such that $x<\min\{p,\hat{x}^u\}$.  Moreover, a sufficient condition for \eqref{eq: IFF_upper}, depending only on the primitive distributions, is
\begin{equation}\label{eq: SC_upper}
    F_2(x)
    \ge
    \frac{\sqrt{\pi_{1,x}}-\sqrt{\pi_{2,x}}}{\sqrt{\pi_{1,x}}+\sqrt{\pi_{2,x}}}\,\Delta^u(x)
    \qquad\text{for all }x<\hat{x}^u.
\end{equation}
\end{proposition}

\begin{example}[Monotone under upper-censoring]
Consider the previous example of  
\begin{equation*}
L_1=(0.25, 0.25, 0.25, 0.25) \quad\text{ and }\quad L_2=(0.2, 0.1, 0.3, 0.4).
\end{equation*}
Observe that $L_{1}\prec _{sc}L_{2}$ and $\hat{x}^u=3$.  For $x=1$, $\Delta^u(1)=0$, and condition \eqref{eq: SC_upper} holds obviously.   For $x=2$, $F_2(2)=0.3$ and $\Delta^u(2)=\sum_{z=1}^{1}\left( \pi_{2,z}-\pi_{1,z}\right) \left( 2-z\right)=\left( \pi_{2,1}-\pi_{1,1}\right) \left( 2-1\right)=-0.05$. It is obvious that $F_2(2)$ is greater than
\begin{equation*}
    \frac{\sqrt{\pi_{2,2}}-\sqrt{\pi_{1,2}}}{\sqrt{\pi_{2,2}}+\sqrt{\pi_{1,2}}}\Delta^u \left( 2\right)=-\frac{\sqrt{0.1}-\sqrt{0.25}}{\sqrt{0.1}+\sqrt{0.25}}\times0.05\simeq 0.0113,
\end{equation*}
and  \eqref{eq: SC_upper} also holds. Therefore, the search problem under upper-censoring is monotone.
\end{example}

\paragraph{The stopping rule.}
In a monotone consumer price-search problem, the optimal stopping rule is myopic. For $p\in K^u$, define the price cutoff for optimal stopping by
$
    p^*(\alpha):=\max\{z\in K^u:U^u(z,\alpha)\le c\}.
$
Then it is optimal to stop if and only if $p\le p^*(\alpha)$. That is, once the current price is sufficiently low, the expected reduction from obtaining one more quote is no greater than the search cost. Since $U^u(p,\alpha)$ is increasing in $\alpha$, the cutoff $p^*(\alpha)$ is weakly decreasing in $\alpha$. More optimistic beliefs that the underlying price distribution is $L_1$ increase the expected benefit of another search, making the consumer willing to continue searching at lower prices. Equivalently, the stopping threshold decreases with $\alpha$.

%%%%%%%%%%%%%%%%%%%%%%%%%%%%%%%%%%
\section{Extensions and Further Discussions}\label{sec: further discussions}
\subsection{Monotonicity under Full Revelation}\label{subsec: full revelation}
We return to the baseline model of Section \ref{sec: model} but now consider the information structure of full revelation.  In this case, for any sample outcome $X=x$, the observed signal is $S^f(x)=x$.  Given a state $(y,\alpha)$ and a sample outcome $X=x$, the updated posterior under full revelation is 
\begin{equation}\label{eq: a_t+1 full}
    \alpha'=\dfrac{\alpha \pi_{1,x}}{\pi_{\alpha}(x)}
\end{equation}
We are still interested in the monotone condition under full revelation, that is, when condition \eqref{def: monotone} holds. Naturally, the monotone conditions under full revelation should be more restrictive than those under censoring.  The formal result is given as follows.

\begin{proposition}\label{prop: monotone full->censoring}
If the search problem is monotone under full revelation, then it is monotone under lower-censoring.
\end{proposition}

The result of Proposition \ref{prop: monotone full->censoring} applies to any censoring or coarsening rule that partitions full-revelation outcomes into observable signal cells. The insight is that the posterior of a coarse signal is a conditional average of the posteriors of the fully revealed outcomes in that cell.  Because $U(y,\alpha)$ is affine in the posterior belief $\alpha$, monotonicity under all fully revealed outcomes implies monotonicity under any form of garbling of those outcomes. Moreover, the result does not rely on any assumption on the stochastic orders between $L_1$ and $L_2$.

Moreover, full revelation weakly increases the DM's search payoff relative to any censored information structure.  This is because, under full revelation, the DM can always ignore the extra information and replicate the strategy used under censoring. Therefore, the value under full revelation is weakly greater than the value under censoring, holding other things equal.

However, full revelation also makes the monotonicity conditions harder to satisfy, as it generates stronger posterior movements, which is obvious by comparing \eqref{eq: a_t+1 full} with \eqref{eq: a_t+1 lower}.  Theorem \ref{thm: IFF full} gives a sharp characterization of monotonicity under full revelation.

\begin{theorem}\label{thm: IFF full}
The search problem with full revelation is monotone if and only if 
\begin{equation}\label{eq: IFF_full}
    L_1\prec_{sc}L_2\text{ and }\hat{x}\ge k-1. 
\end{equation}    
\end{theorem}

 Theorem~\ref{thm: IFF full} serves as a critical theoretical benchmark that reinforces the central premise of this paper: unrestricted Bayesian learning generically destroys the tractability of optimal sequential search. Condition~\eqref{eq: IFF_full} requires that the two candidate distributions satisfy the single-crossing order with the crossing point occurring at the very top of the support ($\hat{x} \ge k-1$). This represents a highly restrictive, stark negative result: full information preserves pathwise monotonicity only when every non-maximal observation weakly shifts the posterior toward the inferior distribution $L_1$, leaving the absolute maximal realization $k$ as the sole outcome capable of inducing optimism. Whenever an intermediate draw below the fallback value conveys a positive learning signal ($\pi_{1,x} < \pi_{2,x}$ for some $x \le \hat{x} < k-1$), full revelation induces an optimistic posterior shift that can dominate the negative fallback-value effect, thereby destroying the simple reservation-cutoff policy. By demonstrating that full-information search fails to be monotone under standard distributional settings, this benchmark directly highlights why one-sided information censoring is not merely a specialized modeling assumption, but a structurally necessary mechanism that stabilizes posterior volatility and restores optimal cutoff rules in Bayesian search environments.
For example, consider
\[
L_1=(0.25, 0.25, 0.25, 0.25)
\quad\text{and}\quad
L_2=(0.2, 0.1, 0.3, 0.4).
\]
Here, $\hat{x}=2<k-1=3$, so condition \eqref{eq: IFF_full} is violated. Consequently, the search problem is not monotone under full revelation, even though it is monotone under both lower- and upper-censoring.

When $k=2$, condition \eqref{eq: IFF_full} is satisfied trivially, implying that the search problem under full revelation is always monotone. Moreover, \eqref{eq: IFF_full} implies that $p_j\ge q_j$ for every $j\le k-1$, so the posterior belief $\alpha_t$ is nondecreasing over time. Therefore, whenever the search problem under full revelation is monotone, the stopping threshold $y^*(\alpha_t)$ is nonincreasing over time.

%%%%%%%%%%%%%%%%%%%%%%%
\subsection{Expected monotonicity}\label{subsec: expected monotonicity}

To clarify the decision-theoretic implications of our results, it is useful to contrast pathwise monotonicity with the weaker concept of \emph{expected monotonicity} (frequently termed supermartingale monotonicity).  The definition is as follows:

\begin{definition}[Expected monotonicity]
The search problem \eqref{eq:Bellman-general} satisfies \emph{expected monotonicity} if, for every feasible state $(y,\boldsymbol{\alpha})$, one has
\begin{equation}\label{def: expected monotone}
U(y,\boldsymbol{\alpha})\ge\E_{\boldsymbol{\alpha}} \left[U(y',\boldsymbol{\alpha}')\right],
\end{equation}    
where $y'$ and $\boldsymbol{\alpha}'$ are given by \eqref{eq: updated state}.
\end{definition}

Expected monotonicity requires that the one-step marginal search value decreases on average.  Economically, this condition formalizes the natural notion of expected diminishing returns to search. On the one hand, the fallback value $y$ weakly increases over time; the bar for payoff improvements constantly rises, driving down the expected benefit of further sampling. On the other hand, by Bayes' plausibility, posterior beliefs cannot systematically drift in a favorable direction on average. Expected monotonicity guarantees that this negative fallback effect dominates the expected learning effect on average.

However, expected monotonicity alone is insufficient to guarantee the optimality of myopic, single-threshold cutoff rules.  Because expected monotonicity allows upward jumps in $U$ on some realizations, it does not prevent a situation where today's one-step gain is negative, i.e., $U(y,\alpha)\le c$, but a favorable future realization could push $U(y',\alpha')>c$, making it optimal to continue even though the immediate draw is not worthwhile.  A rational DM may tolerates a negative immediate expected gain to gamble on reaching a better information state.  Therefore, the optimal stopping rule is generally not myopic.  

Pathwise monotonicity, by contrast, eliminates this option value altogether by requiring
$U(y',\alpha') \le U(y,\alpha)$ along \emph{every} feasible realization path. This realization-by-realization bound ensures that, once the one-step marginal search value falls below the search cost $c$, no subsequent outcome can raise it above $c$. Thus, pathwise monotonicity rules out any continuation in which the gap $U(y',\alpha')-U(y,\alpha)$ fluctuates around zero, thereby ensuring that the simple reservation-cutoff policy remains robustly optimal.

In addition to its expected-value foundation, expected monotonicity offers an important operational advantage. Under pathwise monotonicity, the optimal decision rule must evaluate each realized state update $(y',\alpha')$ and adjust the stopping decision accordingly. When expected monotonicity holds, by contrast, the \textit{expected} marginal gain from continuing the search for an additional period is guaranteed to be non-increasing over time, without requiring this monotonicity to hold along every realized path. Consequently, if circumstances make real-time monitoring impractical or impossible, the fixed stopping time that maximizes the expected search surplus can be determined ex ante by a cutoff rule. Specifically, a non-attending DM, for example, an automated recruitment algorithm, can commit before the search begins to continue searching until the first period at which the expected net gain from one additional period becomes non-positive. Expected monotonicity ensures that this cutoff identifies the fixed stopping time that maximizes the expected search surplus among all predetermined stopping times. This eliminates the administrative and cognitive burden of continuous updating and makes it possible to delegate complex sequential search problems to automated decision mechanisms. Importantly, this predetermined stopping time need not be an optimal stopping time in the full, dynamic sense: it is optimal only within the class of fixed stopping times, and its operational appeal stems from the temporal monotonicity of expected marginal gains.

Expected monotonicity also has a direct economic rationale in settings where decisions are evaluated ex ante or at an aggregate level. First, it provides a natural formulation of diminishing marginal returns to search in expectation. When a DM or firm manages a portfolio of search processes, ex ante optimality depends on expected continuation gains rather than on realization-by-realization bounds that must hold along every path. Second, expected monotonicity captures the fundamental interaction between Bayesian updating and free recall. By Bayes plausibility, posterior beliefs form a martingale, so the expected posterior equals the prior. At the same time, free recall generates a deterministic increase in the fallback value $y$. Expected monotonicity thus reflects the structural dominance of the deterministic increase in the fallback value over the mean-preserving effect of Bayesian updating. This provides a natural and robust justification for search environments in which decisions are based on expected values rather than on pathwise guarantees.

The next theorem provides the necessary and sufficient condition, as well as a state-free sufficient condition, for expected monotonicity.  One important observation is that the expected monotone conditions are identical across the two information structures of lower-censoring and full revelation.  

\begin{theorem}[Expected monotonicity]\label{thm: sm monotone}
Suppose Assumption \ref{ass: single-crossing} holds.  Under either lower-censoring or full revelation, the search problem is expected monotone if and only if, for any state $(y,\alpha)$,
\begin{equation}\label{eq: IFF expected}
\sum_{m=y}^{k-1}\bar F_\alpha(m)^2
\ge
\alpha(1-\alpha)\left[(F_2(y)-F_1(y))\Delta(y)+\sum_{x=y+1}^k
\bigl(\pi_{2,x}-\pi_{1,x}\bigr)\Delta(x)\right].
\end{equation}
A sufficient condition for expected monotonicity that depends only on model primitives is
\begin{equation}\label{eq: SC expected lower}
\sum_{m=0}^{k-1}\bar F_1(m)^2\ge\sum_{x>\hat{x}}
\pi_{2,x}
\frac{\sqrt{\pi_{2,x}}-\sqrt{\pi_{1,x}}}{\sqrt{\pi_{2,x}+\sqrt{\pi_{1,x}}}}
\Delta(x).
\end{equation}

\end{theorem}

Condition \eqref{eq: IFF expected} is the expected-value counterpart of the pathwise monotonicity condition in Theorem \ref{thm: monotone conditions_lower}.  The left hand side of the inequality, $\sum_{m=y}^{k-1}\bar F_\alpha(m)^2$, measures the expected fallback-value effect.  The right-hand side is the expected aggregated learning effect.  Specifically, the first term $\alpha(1-\alpha)(F_2(y)-F_1(y))\Delta(y)$ measures the expected learning effect of censored events, which is negative under assumption \ref{ass: single-crossing}; the second term $\alpha(1-\alpha)\sum_{x=y+1}^k
\bigl(\pi_{2,x}-\pi_{1,x}\bigr)\Delta(x)$ measures the expected learning effect of fully revealed outcomes.

The expected monotone conditions under lower-censoring and full revelation are identical. To see this, consider the decomposition in \eqref{eq: U(y',a')-U(y,a)}. First, both information structures induce the same fallback-value transition, $y'=\max\{y,X\}$. Hence, the expected fallback-value effect
$\mathbb E_\alpha[U(y',\alpha)-U(y,\alpha)]$, is identical under the two regimes. Second, as shown in the proof of Theorem~\ref{thm: sm monotone}, the expected learning effect,
$\mathbb E_\alpha[U(y',\alpha')-U(y',\alpha)]$, is also identical after taking expectations.

The key insight is Bayes plausibility. The one-step search value $U(y,\alpha)$ is affine in the belief $\alpha$. Although lower-censoring and full revelation generate different distributions of the posterior belief $\alpha'$, they have the same ex ante expected posterior belief by Bayes plausibility. Therefore, they induce the same conditional expected learning effect. Intuitively, under lower-censoring, the probability-weighted posterior drift associated with the pooled event $\{X\leq y\}$ is exactly equal to the sum of the probability-weighted posterior drifts associated with the corresponding fully revealed outcomes.

\begin{remark}[Pathwise vs.\ expected monotonicity]\label{rem:pathwise-sm-comparison}
The distinction between expected and pathwise monotonicity under lower-censoring is most transparent from the decomposition
\[
U(y',\alpha')-U(y,\alpha)
=
(\alpha-\alpha')\Delta(\max\{y,X\})
-
\ind_{\{X>y\}}\sum_{m=y}^{X-1}\bar F_\alpha(m).
\]
Pathwise monotonicity requires this realized change to be nonpositive for every feasible realization of $X$. This is equivalent to condition \eqref{eq: IFF_lower}: for every $x>\max\{y,\hat x\}$, it holds that  $(\alpha-\alpha'_x)\Delta(x) \le \sum_{m=y}^{x-1}\bar F_\alpha(m)$.

Expected monotonicity only requires the expectation to be nonpositive, i.e.,
\begin{equation*}
\E_{\alpha}\Big[
\ind_{\{X\le y\}}(\alpha-\alpha'_y)\Delta(y)
+
\ind_{\{X>y\}}( \alpha-\alpha'_X)\Delta(X)
-
\ind_{\{X>y\}}\sum_{m=y}^{X-1}\bar F_\alpha(m)\Big]
\le 0,
\end{equation*}
which is the exact condition \eqref{eq: IFF expected}.  Therefore, pathwise monotonicity under lower-censoring implies expected monotonicity by taking expectations, but the converse does not hold.
%The same distinction explains the difference between the primitive sufficient conditions. The pathwise condition \eqref{eq: SC_lower} controls each high realization $x>\hat x$ separately, whereas a supermartingale sufficient condition controls only the expected positive learning drift. 
\end{remark}

In summary, the relationship among the monotonicity conditions has two aspects. First, the expected monotone conditions under lower-censoring and full revelation are identical. Second, under lower-censoring, the expected monotonicity condition \eqref{eq: IFF expected} is weaker than the pathwise monotonicity condition \eqref{eq: IFF_lower}; the latter is, in turn, weaker than the pathwise monotonicity condition \eqref{eq: IFF_full} under full revelation.

In fact, we can generalize the above results based on two facts: $U(y,\alpha)$ is affine in belief $\alpha$, and the Bayes' plausibility constraint.  For pathwise monotonicity conditions, consider all information structures that fully reveal the better-off sample outcomes, i.e., all the sample outcomes satisfying $X>y$ are fully revealed and then result in the same fallback value transition $y'=\max\{y,X\}$, lower-censoring pools all the non-improving outcomes into a single signal of $\{X\le y\}$, and therefore, is the least informative information structure and then renders the weakest condition for pathwise monotonicity.

\begin{example}[expected monotone but not pathwise monotone]\label{ex:sm-not-pathwise}
Consider
\[
L_1=\left(0.98, 0.01, 0.01\right),
\qquad
L_2=\left(0.64, 0.18, 0.18\right).
\]
Observe that $L_1\prec_{sc}L_2$ and $\hat x=1$. Moreover, by \eqref{eq: Delta(y)}, we have 
\[
\Delta(1)
=
0.17\times(2-1)+0.17\times(3-1)
=
0.51,
\qquad
\Delta(2)=0.17,
\qquad
\Delta(3)=0.
\]
We first show that pathwise monotonicity under lower-censoring fails.  Fix a state $(y,\alpha)=(1,0.9)$ and suppose the next realized outcome is $x=2$. Then by \eqref{eq: a_t+1 lower}, the updated posterior is
\[
\alpha'_2
=
\frac{\alpha\pi_{1,2}}{\pi_\alpha(2)}
=
\frac{0.9\times 0.01}{0.027}
=
\frac{1}{3}<\alpha.
\]
Hence, the learning effect is $(\alpha-\alpha'_2)\Delta(2)=
\frac{289}{3000}$.
The fallback-value effect from $y=1$ to $x=2$ is
\[
\sum_{m=1}^{x-1}\bar F_\alpha(m)
=
\bar F_\alpha(1)
=
\pi_\alpha(2)+\pi_\alpha(3)
=
0.027+0.027
=
0.054<(\alpha-\alpha'_2)\Delta(2).
\]
Condition \eqref{eq: IFF_lower} fails. Therefore, the search problem is not pathwise monotone.

However, the search problem is expected monotone. For any state $(y,\alpha)$, we have 
\[
\mathbb E_\alpha\!\left[
(\alpha-\alpha')\Delta(\max\{y,X\})
\right]
=
\alpha(1-\alpha)
\left[
\bigl(F_2(y)-F_1(y)\bigr)\Delta(y)
+
\sum_{x>y}
(\pi_{2,x}-\pi_{1,x})\Delta(x)
\right].
\]
We verify condition \eqref{eq: IFF expected} for each $y\in\bar{K}$. If $y=0$, all outcomes are revealed, so
\[
\sum_{x>0}(\pi_{2,x}-\pi_{1,x})\Delta(x)
=
-0.34\times 0.51 + 0.17\times 0.17+ 0.17\times 0
%=-\frac{289}{2000}
<0.
\]
If $y=1$, the event $X\le 1$ is censored, and
\[
(F_2(1)-F_1(1))\Delta(1)
+
\sum_{x>1}(\pi_{2,x}-\pi_{1,x})\Delta(x)
=
-0.34\times 0.51 + 0.17\times 0.17 
% =-\frac{289}{2000}
<0.
\]
If $y=2$, then
\[
(F_2(2)-F_1(2))\Delta(2)
+
(\pi_{2,3}-\pi_{1,3})\Delta(3)
=
-0.17\times 0.17 + 0.17\times 0
%=-\frac{289}{10000}
<0.
\]
If $y=3$, then $\Delta(3)=0$, and the expected learning effect is zero.  
Therefore, for every $y\in\bar K$ and every $\alpha\in[0,1]$, the expected learning effect $\mathbb E_\alpha\!\left[(\alpha-\alpha')\Delta(\max\{y,X\})\right]\le 0$, and the expected monotone condition \eqref{eq: IFF expected} holds for every state $(y,\alpha)$. Hence, the search problem is expected monotone but not pathwise monotone under lower-censoring.
\end{example}

 \section{Final Remarks}\label{sec: conclusion}
\subsection{The related dynamic programming problem}\label{subsec: DP}

Let us reconsider the baseline model in Section \ref{sec: model}, now allowing for $\ell$ possible underlying distributions, so that
$
L\in\{L_1,\ldots,L_\ell\}.
$
We consider several relevant information structures, including lower-censoring and full revelation.  To analyze \eqref{eq:Bellman-general}, it is useful to define the Bellman operator $T$ on functions over the state space
$
\bar K\times\Delta(\{1,\ldots,\ell\})
$
by
\begin{equation}\label{eq:operator-general}
Tf(y,\boldsymbol{\alpha})
=
\max\left\{
    y,\,
    \E_{\boldsymbol{\alpha}}\left[
        f\bigl(y,B(\boldsymbol{\alpha},X)\bigr)
    \right]-c
\right\}.
\end{equation}
The operator $T$ maps functions on $\bar K\times\Delta(\{1,\ldots,\ell\})$ into themselves. Consequently, any solution $J^r$ of the Bellman equation \eqref{eq:Bellman-general} is a fixed point of $T$:
$
J^r=TJ^r.
$

Establishing the existence and, more importantly, the appropriate characterization of such a fixed point is nontrivial. The search problem is undiscounted, and the Bellman operator is monotone but generally not a contraction. Thus, standard contraction-based arguments do not apply directly. Moreover, the Bellman equation may admit multiple fixed points, not all of which correspond to optimal behavior.

To illustrate this point, consider the constant function
$
J^r(y,\boldsymbol{\alpha})=k
$
for every state $(y,\boldsymbol{\alpha})$, where $k$ is the maximal attainable payoff. This function can satisfy the Bellman equation even though the associated behavior need not be optimal: the induced strategy may continue searching whenever $y<k$, regardless of the posterior belief or of the likelihood of obtaining an improvement. Thus, the relevant question is not merely whether the Bellman operator has a fixed point, but which fixed point corresponds to the optimal value function.

This issue has been studied in the literature on adaptive sequential search. \cite{rosenfield1981optimal}, for example, formulate an adaptive sequential search problem through a functional Bellman equation. Their analysis focuses primarily on deriving structural conditions on the primitive distributions that guarantee a monotone reservation-price policy. In particular, their approach imposes conditions ensuring that posterior beliefs respond monotonically to observed realizations and that the effect of learning does not overwhelm the effect of changes in the fallback value.

A general treatment of this fixed-point problem is provided by \cite{Bertsekas2013}. In the context of undiscounted dynamic programming, Bertsekas characterizes the optimal value function, and thereby the optimal stopping rule, as the smallest fixed point of the Bellman operator.

The general approach of \cite{Bertsekas2013}, however, involves considerable technical machinery. This motivates our more elementary treatment. In the Online Appendix, for completeness, we provide a self-contained proof of the relevant result of \cite{Bertsekas2013}, tailored to the sequential search problems considered in this paper. Our proof adapts standard convergence arguments from undiscounted dynamic programming to the present setting. It also yields a transparent connection between the finite-horizon search problems and their optimal stopping rules, on the one hand, and the corresponding infinite-horizon problem and its optimal stopping rule, on the other.

The main drawback of this approach is its technical complexity.  This limitation motivates the broader objective of the paper

\subsection{Conclusion}
This paper studies Bayesian sequential search when observations are endogenously censored by the DM's current fallback option. The central insight is that one-sided censoring acts as an information-stabilizing force, restoring the monotonicity that unrestricted Bayesian learning typically destroys. This stabilization is captured by decomposing the intertemporal change in the one-step marginal search value into a negative fallback-value effect and an ambiguous learning effect. We demonstrate that under suitable single-crossing conditions, lower-censoring mathematically bounds the magnitude of positive belief shifts, ensuring the fallback-value effect dominates along all feasible search histories. 

The analysis reveals a strict tractability hierarchy across information structures. Pathwise monotonicity under full revelation requires highly restrictive conditions, specifically that the single-crossing point lies at the top of the support. Lower-censoring preserves pathwise monotonicity, and thus the optimality of simple reservation-cutoff policies, under strictly weaker primitive conditions. Furthermore, we show that expected monotonicity depends solely on the conditional expected learning effect. Because the marginal search value is affine in the posterior belief, Bayes plausibility dictates that the expected monotone conditions are identical across both lower-censoring and full revelation. Consequently, informational coarsening preserves tractability exclusively by suppressing realized posterior volatility, not by altering the average drift of beliefs.

In addition to characterizing monotonicity, the paper provides a rigorous dynamic-programming foundation for general Bayesian search environments. Because the sequential search problem is undiscounted, the associated Bellman operator lacks standard contraction properties, and the Bellman equation generally admits multiple fixed points. To resolve this multiplicity, we demonstrate that the economically relevant value function can be constructed as the pointwise limit of finite-horizon value functions. The resulting function is uniquely identified as the lowest fixed point of the Bellman operator, and it is this minimal solution that dictates the optimal stopping strategy for the infinite-horizon problem. This constructive approach provides a formal justification for the use of the Bellman equation while keeping the existence and verification of optimal policies conceptually distinct from the structural conditions required for monotonicity.

Several theoretical and applied extensions remain open. First, extending the framework to continuous payoff distributions or richer multidimensional state spaces requires moving beyond scalar single-crossing conditions. In such environments, the learning effect becomes multidimensional, and preserving monotonicity will likely demand strict shape restrictions on the posterior predictive distributions, such as log-concavity or the Monotone Likelihood Ratio Property (MLRP). Second, the information structure could be endogenized within a formal mechanism design or Bayesian persuasion framework. Rather than assuming an exogenous censoring rule, future research could analyze how a principal, such as a hiring platform or a market intermediary, optimally designs the censoring threshold to maximize the searcher's engagement or the aggregate match quality. Finally, the theoretical separation between pathwise and expected monotonicity provides a structural foundation for empirical work. The model generates precise, testable implications regarding how DMs ought to adjust their reservation cutoffs following coarse failures, such as uninformative job-market rejections or opaque price screening. Testing these theoretical bounds against empirical belief-updating behavior will deepen our understanding of when informational coarsening serves strictly to bound learning, and when it acts as a necessary condition for the tractability of sequential choice.

\bibliography{SearchLearning}{}

@article{adam2001learning,
  title={Learning while searching for the best alternative},
  author={Adam, Klaus},
  journal={Journal of Economic Theory},
  volume={101},
  number={1},
  pages={252--280},
  year={2001},
  publisher={Elsevier}
}

@article{baucells2025search,
  title={Search in the dark: The case with recall and gaussian learning},
  author={Baucells, Manel and Zorc, Sa{\v{s}}a},
  journal={Operations Research},
  volume={73},
  number={5},
  pages={2572--2590},
  year={2025},
  publisher={INFORMS}
}

@incollection{bergemann2018bandit,
  title={Bandit problems},
  author={Bergemann, Dirk and V{\"a}lim{\"a}ki, Juuso},
  booktitle={The new Palgrave dictionary of economics},
  pages={665--670},
  year={2018},
  publisher={Springer}
}

@article{bikhchandani1996optimal,
  title={Optimal search with learning},
  author={Bikhchandani, Sushil and Sharma, Sunil},
  journal={Journal of Economic Dynamics and Control},
  volume={20},
  number={1-3},
  pages={333--359},
  year={1996},
  publisher={Elsevier}
}

@article{degroot1968some,
  title={Some problems of optimal stopping},
  author={DeGroot, Morris H},
  journal={Journal of the Royal Statistical Society Series B: Statistical Methodology},
  volume={30},
  number={1},
  pages={108--122},
  year={1968},
  publisher={Oxford University Press}
}

@article{ferguson1979bayesian,
  title={Bayesian nonparametric estimation based on censored data},
  author={Ferguson, Thomas S and Phadia, Eswar G},
  journal={The annals of statistics},
  pages={163--186},
  year={1979},
  publisher={JSTOR}
}

@unpublished{fershtman2025searching,
  title={Searching for “Arms”: Experimentation with Endogenous Consideration Sets},
  author={Fershtman, Daniel and Pavan, Alessandro},
  note={Working paper},
  year={2025}
}

@article{gittins1974dynamic,
  title={A dynamic allocation index for the sequential design of experiments},
  author={Gittins, John},
  journal={Progress in statistics},
  pages={241--266},
  year={1974},
  publisher={North Holland}
}

@article{he2022mislearning,
  title={Mislearning from censored data: The gambler's fallacy and other correlational mistakes in optimal-stopping problems},
  author={He, Kevin},
  journal={Theoretical Economics},
  volume={17},
  number={3},
  pages={1269--1312},
  year={2022},
  publisher={Wiley Online Library}
}

@article{hutter2013bayesian,
  title={Bayesian optimization with censored response data},
  author={Hutter, Frank and Hoos, Holger and Leyton-Brown, Kevin},
  journal={arXiv preprint arXiv:1310.1947},
  year={2013}
}

@article{koulayev2013search,
  title={Search with dirichlet priors: estimation and implications for consumer demand},
  author={Koulayev, Sergei},
  journal={Journal of Business \& Economic Statistics},
  volume={31},
  number={2},
  pages={226--239},
  year={2013},
  publisher={Taylor \& Francis}
}

@article{mccall1970economics,
  title={Economics of information and job search},
  author={McCall, John Joseph},
  journal={The Quarterly Journal of Economics},
  volume={84},
  number={1},
  pages={113--126},
  year={1970},
  publisher={MIT Press}
}

@article{preuss2023search,
  title={Search, learning, and tracking},
  author={Preuss, Marcel},
  journal={The RAND Journal of Economics},
  volume={54},
  number={1},
  pages={54--82},
  year={2023},
  publisher={Wiley Online Library}
}

@article{smith2017risk,
  title={Risk aversion, information acquisition, and technology adoption},
  author={Smith, James E and Ulu, Canan},
  journal={Operations Research},
  volume={65},
  number={4},
  pages={1011--1028},
  year={2017},
  publisher={INFORMS}
}

@article{santos2017search,
  title={Search with learning for differentiated products: Evidence from e-commerce},
  author={Santos, Babur De Los and Horta{\c{c}}su, Ali and Wildenbeest, Matthijs R},
  journal={Journal of Business \& Economic Statistics},
  volume={35},
  number={4},
  pages={626--641},
  year={2017},
  publisher={Taylor \& Francis}
}

@article{talmain1992search,
  title={Search from an unkown distribution an explicit solution},
  author={Talmain, Gabriel},
  journal={Journal of Economic Theory},
  volume={57},
  number={1},
  pages={141--157},
  year={1992},
  publisher={Elsevier}
}

@article{ulu2009uncertainty,
  title={Uncertainty, information acquisition, and technology adoption},
  author={Ulu, Canan and Smith, James E},
  journal={Operations Research},
  volume={57},
  number={3},
  pages={740--752},
  year={2009},
  publisher={INFORMS}
}

@article{weitzman1979optimal,
	title={Optimal search for the best alternative},
	author={Weitzman, Martin L},
	journal={Econometrica: Journal of the Econometric Society},
	pages={641--654},
	year={1979}
}

@article{xu2025search,
  title={Search without Recall and Gaussian Learning: Structural Properties and Optimal Policies},
  author={Xu, Stephen Z and Baucells, Manel and Zorc, Sa{\v{s}}a},
  journal={Available at SSRN 5414554},
  year={2025}
}

@article{rosenfield1981optimal,
  author  = {Rosenfield, Donald B. and Shapiro, Roy D.},
  title   = {Optimal Adaptive Price Search},
  journal = {Journal of Economic Theory},
  volume  = {25},
  number  = {1},
  pages   = {1--20},
  year    = {1981},
  month   = aug,
  doi     = {10.1016/0022-0531(81)90014-4}
}

@book{Bertsekas2013,
  author    = {Bertsekas, Dimitri P.},
  title     = {Abstract Dynamic Programming},
  publisher = {Athena Scientific},
  year      = {2013},
  address   = {Belmont, MA}
}
\bibliographystyle{ecta}

\appendix 

\section*{Appendix: Proofs}\label{app:proof}
%------------------
\begin{proof}[Proof of Proposition \ref{prop: optimal stopping monotone}]
First, at a current state $(y,\boldsymbol\alpha)$, if $U(y,\boldsymbol\alpha)>c$, an immediate stopping is not optimal obviously.  Next, suppose $U(y,\boldsymbol\alpha)\le c$. Consider any continuation policy that starts from $(y,\boldsymbol\alpha)$ and stops after some future number of inspections. For notational simplicity, first consider policies that stop within at most $T$ additional inspections. Let $(Y_s,\boldsymbol\alpha_s)$ denote the state after $s$ additional inspections, with $(Y_0,\boldsymbol\alpha_0)=(y,\boldsymbol\alpha)$, and let $\tau\le T$ be the induced stopping time, measured from the current state.

Since the problem is monotone and $U(y,\boldsymbol\alpha)\le c$, every continuation path satisfies
\[
    U(Y_s,\boldsymbol\alpha_s)\le c
    \qquad\text{for all }s\ge 0.
\]
Therefore, the expected net gain from following the continuation policy is
\begin{equation*}
    \E\bigl[Y_\tau-y-c\tau\bigr]
    = \E\left[\sum_{s=0}^{\tau-1}(Y_{s+1}-Y_s-c)\right]
    = \E\left[\sum_{s=0}^{\tau-1}\bigl(U(Y_s,\boldsymbol\alpha_s)-c\bigr)\right]
    \le 0.
\end{equation*}
Thus, no bounded continuation policy improves on immediate stopping. The extension to admissible unbounded stopping policies follows by the standard truncation argument, provided expected payoffs are well defined. Hence stopping is optimal whenever $U(y,\boldsymbol\alpha)\le c$.
\end{proof}

%-----------------
\begin{proof}[Proof of Lemma \ref{lem: st order}]
For any $\alpha,\alpha'\in[0,1]$ and $j\in K$, we have
\begin{equation*}
    \pi_{\alpha'}(j)-\pi_{\alpha}(j)=(\alpha'-\alpha)(\pi_{1,j}-\pi_{2,j}).
\end{equation*}
The claims for the $sc$, $st$, and $icx$ orders follow directly from the linearity of $\pi_\alpha$, $F_\alpha$, and the integrated survival functions in $\alpha$.  
\ignore{
For the $lr$ order, let
\begin{equation*}
\eta_{j}\equiv \frac{\pi_{\alpha'}(j) }{\pi_{\alpha}(j) }=\frac{\alpha'\cdot \left( \pi_{1,j}-\pi_{2,j}\right)
+\pi_{2,j}}{\alpha\cdot \left( \pi_{1,j}-\pi_{2,j}\right) +\pi_{2,j}}
\end{equation*}%
and the result is implied by $\eta_{j+1}-\eta_{j}\propto (\alpha'-\alpha) \left( \frac{\pi_{1,j+1}}{\pi_{2,j+1}}-\frac{\pi_{1,j}}{\pi_{2,j}}\right) $, where the sign follows from the likelihood-ratio ordering.
}
\end{proof}

%------------------
\begin{proof}[Proof of Lemma \ref{lem: Delta(y)}]
(i) As Assumption \ref{ass: single-crossing} implies $L_1\prec_{st}L_2$, then $F_1(y)\ge F_2(y)$ for all $y\in\bar{K}$.  (ii) As $L_1\prec_{st}L_2$ implies $L_1\prec_{icx}L_2$, then $\Delta(y)\ge 0$ for all $y\in\bar{K}$. Moreover, $\Delta(k)=0$ because $(X-k)_+=0$ for $X\in K$. (iii) Observe that
\[
    \Delta(y)-\Delta(y+1)
    =
    \sum_{z=y+1}^k(\pi_{2,z}-\pi_{1,x})
    =
    F_1(y)-F_2(y)
    \ge0.
\]
Therefore $\Delta(y)$ is weakly decreasing in $y$.
\end{proof}

%---------------------
\begin{proof}[Proof of Lemma \ref{lem: U(y,a)}]
(i) For any $\alpha \in \left[ 0,1\right] $, we have
\begin{equation*}
U\left( y+1,\alpha \right) -U\left( y,\alpha \right) %=\sum_{z=y+1}^{k}\pi_{\alpha}(z) \left( z-y-1\right) -\sum_{z=y}^{k}\pi_{\alpha}(z) \left( z-y\right)
=-\sum_{z=y+1}^{k}\pi_{\alpha}(z)
=-\bar{F}_{\alpha}( y)\le 0,
\end{equation*}
with strict inequality for $y<k$ by full support assumption.  %Then for any $y'>y$, \begin{equation*}U\left( y',\alpha \right) -U\left( y,\alpha \right)=\sum_{z=y}^{y'-1}\left[ U\left( z+1,\alpha \right) -U\left(z,\alpha \right) \right] =-\sum_{z=y}^{y'-1}\bar{F}_{\alpha}(z) <0. \end{equation*}
%\footnote{ An alternative expression is \begin{equation*} U\left( y',\alpha \right) -U\left( y,\alpha \right)=\sum_{z=y'}^{k}\pi_{\alpha}(z) \left( z-y^{\prime}\right) -\sum_{z=y}^{k}\pi_{\alpha}(z) \left( z-y\right)=-\left[ \sum_{z=y}^{y'-1}\pi_{\alpha}(z) \left(z-y\right) +\sum_{z=y'}^{k}\pi_{\alpha}(z) \left(y'-y\right) \right]. \end{equation*} }
Moreover, whenever the second-order difference is defined, e.g., for $y=0,\ldots, k-2$, we have
\begin{equation*}
\left[ U\left( y+2,\alpha \right) -U\left( y+1,\alpha \right) \right] -\left[
U\left( y+1,\alpha \right) -U\left( y,\alpha \right) \right] =\bar{F}_{\alpha}(y) -\bar{F}_{\alpha}(y+1)
=\pi_{\alpha}(y+1) >0.
\end{equation*}

(ii)  For $\alpha'>\alpha$, we have $U(y,\alpha')-U(y,\alpha)=(\alpha-\alpha')\Delta(y)$.
By Lemma \ref{lem: Delta(y)}, $\Delta(y)\ge0$, and hence $U(y,\alpha)$ is decreasing in $\alpha$.  Finally, for $y'>y$ and $\alpha'>\alpha$,
{ \setlength{\abovedisplayskip}{3pt}
\setlength{\belowdisplayskip}{3pt}
\[
\bigl[U(y',\alpha')-U(y,\alpha')\bigr]
-
\bigl[U(y',\alpha)-U(y,\alpha)\bigr]
=
(\alpha'-\alpha)\bigl[\Delta(y)-\Delta(y')\bigr]
\ge0,
\]
}
because $\Delta$ is weakly decreasing. Then $U(y,\alpha)$ has increasing differences in $(y,\alpha)$.
\ignore{
Specifically,
\begin{equation*}
\Delta \left( y'\right) -\Delta \left( y\right)
=\sum_{z=y}^{y'-1}\left[ \bar{F}_{1}\left( z\right) -\bar{F}_{2}\left( z\right)
\right]=\sum_{z=y}^{y'-1}\left[ F_{2}\left( z\right) -F_{1}\left( z\right) %
\right] \leq 0.
\end{equation*}
}
\end{proof}

%---------------------

\begin{proof}[Proof of Theorem \ref{thm: monotone conditions_lower}]
For any given state $\left( y,\alpha \right) $ and any search outcome $X=x$, the new fallback value $y'=\max \{x,y\}$ and the new belief $\alpha'$ is derived by (\ref{eq: a_t+1 lower}).  Under lower-censoring, the posterior shift is
\begin{equation}\label{eq: a'-a_lower}
\alpha'-\alpha=
\begin{cases}
\displaystyle
\frac{\alpha(1-\alpha)}{F_{\alpha}(y)}\bigl[F_1(y)-F_2(y)\bigr],
& \text{if } \{X\le y\} \text{ is observed},\\[1em]
\displaystyle
\frac{\alpha(1-\alpha)}{\pi_{\alpha}(x)}(\pi_{1,x}-\pi_{2,x}),
& \text{if } X=x>y \text{ is observed}.
\end{cases}
\end{equation}
We are interested in the conditions for \eqref{def: monotone} to hold, and there are two cases.

\paragraph{Automatic cases.}
First, suppose the observation is censored, i.e., $\{X\le y\}$. Then $y'=y$ and $\alpha'\ge \alpha$ from \eqref{eq: a'-a_lower}.  Hence $
U(y',\alpha')-U(y,\alpha) = U(y,\alpha')-U(y,\alpha) = (\alpha-\alpha')\Delta(y)\le 0$ as $\Delta(y)\ge 0$ by Lemma \ref{lem: Delta(y)}.  Next, suppose the observation is fully revealed with $X=x>y$ and $x\le \hat{x}$. Then $\alpha'\ge \alpha$ as $\pi_{1,x}-\pi_{2,x}\ge 0$, and the learning effect $(\alpha-\alpha')\Delta(x)\le 0$.  As the fallback value effect is negative, $U(x,\alpha')-U(y,\alpha)\le 0$ as well.  So monotonicity holds automatically in both these cases.

\paragraph{The critical case.}
Now consider a fully revealed observation $X=x>\max\{y,\hat{x}\}$. Then $\pi_{1,x}-\pi_{2,x}<0$, and $\alpha'<\alpha$ from \eqref{eq: a'-a_lower}. The change in the one-step marginal search value becomes
\[
U(x,\alpha')-U(y,\alpha)
=
-\sum_{z=y}^{x-1}\bar F_\alpha(z)
+
(\alpha-\alpha')\Delta(x).
\]
The first term (fallback value effect) is strictly negative, while the second term (learning effect) is strictly positive. Thus the sign of the total change is determined by their relative magnitudes.
\begin{itemize}
    \item \textbf{Sufficiency}.  Suppose condition \eqref{eq: IFF_lower} holds for every $(y,\alpha)$ and every $x>\max\{y,\hat{x}\}$. Then for the critical case, $U(x,\alpha')-U(y,\alpha)\le 0$. Combined with the automatic cases, we have $U(y',\alpha')\le U(y,\alpha)$ for every feasible transition. Therefore the search problem is monotone.  
    \item \textbf{Necessity}.  Conversely, suppose the search problem is monotone. Fix any state $(y, \alpha)$, and any $x>\max\{y,\hat{x}\}$. Monotonicity requires $U(x,\alpha')- U(y,\alpha)\le 0$.  Substituting the expression from the critical case, this inequality is equivalent to \eqref{eq: IFF_lower}. This proves the necessity.
\end{itemize}

We next show the sufficient condition \eqref{eq: SC_lower} implies \eqref{eq: IFF_lower}. Consider any $x>\max\{y,\hat{x}\}$. First,
\begin{equation*}
\sum_{z=y}^{x-1}\bar{F}_{\alpha}(z) \geq \bar{F}_{\alpha}(y)
=1-F_{\alpha}(y) \geq 1-F_{1}\left( y\right) \geq
1-F_{1}\left( x-1\right)=\bar{F}_1(x-1),
\end{equation*}% 
where $F_{\alpha}(y)  \leq F_{1}\left( y\right) $ because
$F_{1}\left( y\right) \geq F_{2}\left( y\right)$ by Assumption \ref{ass: single-crossing}.  Second, observe that $r_{x}:=\frac{\pi_{1,x}}{\pi_{2,x}}<1$ for all $x>\max\{y,\hat{x}\}$, and
we have
\begin{equation*}
\alpha -\alpha'=\frac{\alpha \left( 1-\alpha \right) }{\alpha
r_{x}+\left( 1-\alpha \right) }\left( 1-r_{x}\right)
\end{equation*}
achieves its maximum at $\alpha ^{\ast }=\left( 1+\sqrt{r_{x}}\right)
^{-1}$.\footnote{Solving the first order condition of $\frac{\partial }{\partial \alpha }\left[ \frac{\alpha \left( 1-\alpha \right) }{\alpha r_{x}+\left( 1-\alpha\right) }\right] =\frac{1-2\alpha +\left( 1-r_{x}\right) \alpha ^{2}}{\left[\alpha r_{x}+\left( 1-\alpha \right) \right] ^{2}}=0$ gives $\alpha ^{\ast }$.}
 It then follows that%\footnote{Note that we do not have the monotonicity of $\frac{1-\sqrt{r_{x}}}{1+\sqrt{r_{x}}}\Delta \left( x\right) $, as $\frac{1-\sqrt{r_{x}}}{1+\sqrt{r_{x}}}$is increasing in $x$ while $\Delta \left( x\right) $ is decreasing in $x$.}
\begin{equation*}
\left( \alpha -\alpha'\right) \Delta \left( x\right) \leq \frac{
\alpha ^{\ast }\left( 1-\alpha ^{\ast }\right) }{\alpha ^{\ast }r_{x}+\left(
1-\alpha ^{\ast }\right) }\left( 1-r_{x}\right) \Delta \left( x\right) =
\frac{1-\sqrt{r_{x}}}{1+\sqrt{r_{x}}}\Delta \left( x\right)=\frac{\sqrt{\pi_{2,x}}-\sqrt{\pi_{1,x}}}{\sqrt{\pi_{2,x}}+\sqrt{\pi_{1,x}}}\,\Delta(x).
\end{equation*}

We then have the state-free sufficient condition (\ref{eq: SC_lower}) for monotonicity.
\end{proof}

%---------------------
\begin{proof}[Proof of Lemma \ref{lem: y*(a)}]
(i) By Lemma \ref{lem: Delta(y)}, $U_2(y)\ge U_1(y)$ for every $y$, and hence
\[
    U_1(y)\le U(y,\alpha)\le U_2(y)
    \qquad\text{for all }y\in\bar K.
\]
Therefore, the cutoff for $U(y,\alpha)$ lies between the two degenerate-belief cutoffs: $y_1^*\le y^*(\alpha)\le y_2^*$.  (ii) By Lemma \ref{lem: U(y,a)}, $U(y,\alpha)$ is weakly decreasing in $\alpha$. Thus, if $\alpha'\ge\alpha$, then $U(y,\alpha')\le U(y,\alpha)$ for every $y$, and $y^*(\alpha')\le y^*(\alpha)$.  Thus $y^*(\alpha)$ is weakly decreasing in $\alpha$.  Finally, if $c'\ge c$, then $\{y:U(y,\alpha)\le c\}\subseteq\{y:U(y,\alpha)\le c'\}$.  Therefore $y^*(\alpha;c')\le y^*(\alpha;c)$. Hence, the cutoff is weakly decreasing in the search cost $c$.  (iii) The result follows directly from \eqref{eq: belief shift_lower}.
\end{proof}

%---------------------
\begin{proof}[Proof of Lemma \ref{lem:belief-cutoff-stopping}]
$U(y,\alpha)=\alpha U_1(y)+(1-\alpha)U_2(y)$ is decreasing in $\alpha$ by Assumption \ref{ass: single-crossing}.  Moreover, for any $y\in\mathcal{Y}$, $U(y,1)=U_1(y)\le c$, and the set $\{\alpha:U(y,\alpha)\le c\}$ is nonempty. Therefore $\underline\alpha(y)$ is well-defined and satisfies: $\alpha\ge\underline\alpha(y)\Leftrightarrow U(y,\alpha)\le c$.  By Proposition \ref{prop: optimal stopping monotone}, monotonicity implies that stopping is optimal if and only if $U(y,\alpha)\le c$. Hence stopping is optimal if and only if $\alpha\ge\underline\alpha(y)$.
\end{proof}

%---------------------
\begin{proof}[Proof of Proposition \ref{prop:finite-pathwise-bound}]
If $y\ge y_2^*$, then stopping is optimal for every belief, so $\bar n(y,\alpha)=0$.

Now fix $y\in\mathcal Y$ with $y<y_2^*$. We argue by backward induction on $y$. Under Assumption \ref{ass:strict-lowest-cutoff}, $\bar O(y)<\infty$. Since $R(y)>1$, $m_y(\alpha)<\infty$. If the next $m_y(\alpha)$ observations are all censored, the fallback value remains $y$, the posterior odds become $O(\alpha)R(y)^{m_y(\alpha)}$, and by definition of $m_y(\alpha)$, the state reaches the stopping region. Hence, the policy stops within $m_y(\alpha)$ searches along any path with $m_y(\alpha)$ consecutive censored observations.

Otherwise, the first uncensored observation occurs after $m<m_y(\alpha)$ censored observations. Let this revealed value be $x>y$. Immediately before observing $x$, the posterior odds are $O(\alpha)R(y)^m$. After observing $x$, the posterior belief is $\alpha^{(m,x)}$, as given by \eqref{eq: alpha(m,x)}, where
\[
    O(\alpha^{(m,x)})
    =
    O(\alpha)R(y)^m\frac{\pi_{1,x}}{\pi_{2,x}}.
\]
The fallback value becomes $x>y$. If $x\ge y_2^*$, stopping is immediate. If $x<y_2^*$, then by the induction hypothesis, starting from $(x,\alpha^{(m,x)})$, the policy stops within $\bar n(x,\alpha^{(m,x)})$ additional searches. Therefore, along such a path, the total number of additional searches is bounded by
\[
    m+1+\bar n(x,\alpha^{(m,x)}).
\]
Taking the maximum over $m=0,\ldots,m_y(\alpha)-1$ and $x\in\{y+1,\ldots,k\}$, and comparing with the all-censored case, gives the stated recursive bound \eqref{eq: pathwise bound}.  The bound is finite because $m_y(\alpha)<\infty$ and the support $K$ is finite. This completes the induction.

It remains to show that $\bar n(y,\alpha)$ is decreasing in both $\alpha$ and $c$. We prove the property of $\alpha$ by backward induction, and the proof for $c$ is similar.  For the base case $y\ge y_2^*$, we have $\bar n(y,\alpha)=0$ for all $\alpha$ and $c$, which is trivially decreasing in both arguments.

For the inductive step, fix $y<y_2^*$ and assume the monotonicity claim holds for every $x>y$. An increase in $\alpha$ makes stopping strictly more attractive.  We exploit two key monotonic effects.  1) From the definition of the consecutive censoring bound,
\[
m_y(\alpha) := \min\{m\in\mathbb N_0: O(\alpha)R(y)^m\ge \bar O(y)\}.
\]
It follows that $m_y(\alpha)$ is weakly decreasing in $\alpha$: a higher $\alpha$ raises the odds ratio $O(\alpha)$, while a higher $c$ lowers the odds threshold $\bar O(y)$. 2) For any $x>y$ and any $m<m_y(\alpha)$, the posterior odds after $m$ censored draws followed by a revelation of $x$ are
\[
O(\alpha^{(m,x)}) = O(\alpha)R(y)^m\frac{\pi_{1,x}}{\pi_{2,x}},
\]
which is increasing in $\alpha$ and independent of $c$. Thus, by the induction hypothesis, the inner term of $m+1+\bar n\bigl(x,\alpha^{(m,x)}(\alpha)\bigr)$ is weakly decreasing in $\alpha$ wherever it is defined.

Now, as $\alpha$ increases, the feasible set of $m$'s over which the inner maximum in
\[
\bar n(y,\alpha)
=
\max\left\{
    m_y(\alpha),
    \max_{0\le m<m_y(\alpha)}
    \max_{x>y}
    \left[
        m+1+\bar n\bigl(x,\alpha^{(m,x)}(\alpha)\bigr)
    \right]
\right\}
\]
gets smaller (or remains the same), while every term in that maximum is itself weakly smaller. Consequently, the entire expression $\bar n(y,\alpha)$ is weakly decreasing in $\alpha$. This establishes the desired monotonicity simultaneously for both $\alpha$, completing the induction.  The proof for $c$ is similar.
\end{proof}

%---------------------------
\begin{proof}[Proof of Theorem \ref{thm: monotone conditions_lower_continuous}]
For any given state $(y,\alpha)$ and any search outcome $X=x$, the new fallback value $y'=\max \{x,y\}$ and the new belief $\alpha'$ satisfies 
\begin{equation*}
\alpha'-\alpha=
\begin{cases}
\displaystyle
\frac{\alpha(1-\alpha)}{F_{\alpha}(y)}\bigl[F_1(y)-F_2(y)\bigr],
& \text{if } \{X\le y\} \text{ is observed},\\[1em]
\displaystyle
\frac{\alpha(1-\alpha)}{f_{\alpha}(x)}(f_1(x)-f_2(x)),
& \text{if } X=x>y \text{ is observed}.
\end{cases}
\end{equation*}
Consider the monotone condition; there are two cases.

\paragraph{Automatic cases.}
First, suppose the observation is censored, i.e., $\{X\le y\}$. Then $y'=y$ and $\alpha'\ge \alpha$. Hence $
U(y',\alpha')-U(y,\alpha) = (\alpha-\alpha')\Delta(y)\le 0$.  Next, suppose $X=x>y$ and $x\le \hat{x}$. Then $\alpha'\ge \alpha$, and the learning effect $(\alpha-\alpha')\Delta(x)$ is negative, and the fallback value effect is always negative.  So monotonicity holds automatically in both these cases.

\paragraph{The critical case.}
Now consider a fully revealed observation $X=x>\max\{y,\hat{x}\}$. Then $\pi_{1,x}-\pi_{2,x}<0$, so $\alpha'<\alpha$. The change in the marginal search value becomes
\[
U(x,\alpha')-U(y,\alpha)
=
-\int_y^x\bar{F}_{\a}(z)dz
+
(\alpha-\alpha')\Delta(x).
\]
The first term (fallback value effect) is strictly negative, while the second term (learning effect) is strictly positive. Thus the sign of the total change is determined by their relative magnitudes.
\begin{itemize}
    \item \textbf{Sufficiency}.  Suppose condition \eqref{eq: IFF_lower_continuous} holds for every $(y,\alpha)$ and every $x>\max\{y,\hat{x}\}$. Then for the critical case, $U(x,\alpha')-U(y,\alpha)\le 0$. Combined with the automatic cases, we have $U(y',\alpha')\le U(y,\alpha)$ for every feasible transition. Therefore, the search problem is monotone.  
    \item \textbf{Necessity}.  Conversely, suppose the search problem is monotone. Fix any state $(y,\alpha)$ and any $x>\max\{y,\hat{x}\}$. Monotonicity requires $U(x,\alpha')-U(y,\alpha)\le 0$.  Substituting the expression from the critical case, this inequality is equivalent to \eqref{eq: IFF_lower_continuous}. This proves the necessity.
\end{itemize}
\end{proof}

%---------------------------
\begin{proof}[Proof of Lemma \ref{lem: Delta(p) U(p,a)_upper}]
(1) $L_{1}\prec _{sc}L_{2}$ implies $L_{1}\prec _{st}L_{2}$.  Using summation by part, from \eqref{eq:Delta-upper}
\[
\Delta^u(p)
=
-\sum_{z=1}^{p-1}\bigl[F_1(z)-F_2(z)\bigr]\le 0,
\]
and $\Delta^u(p+1)-\Delta^u(p)=F_2(p)-F_1(p)\le0$.
(2) First, for monotonicity, we have
\begin{equation*}
U^u\left( p+1,\alpha \right) -U^u\left( p,\alpha \right)
=\sum_{z=1}^{p}\pi_{\alpha}(z) \left( p+1-z\right) -\sum_{z=1}^{p-1}\pi_{\alpha}(z) \left( p-z\right)
=\sum_{z=1}^{p}\pi_{\alpha}(z) =F_{\alpha}(p)\ge 0.
\end{equation*}%
Whenever the second-order difference exists, we have
\begin{equation*}
\left[ U^u\left( p+2,\alpha \right) -U^u\left( p+1,\alpha \right) \right] -\left[
U^u\left( p+1,\alpha \right) -U^u\left( p,\alpha \right) \right] =F_{\alpha}(
p+1) -F_{\alpha}(p) =\pi_{\alpha}(p+1) >0\text{.}
\end{equation*}
Second, for $\alpha'>\alpha$, we have $U^u(p,\alpha')-U^u(p,\alpha)=(\alpha-\alpha')\Delta^u(p)\ge 0$ because $\Delta^u(p)\le 0$.  Finally, for $p'>p$ and $\alpha'>\alpha$,
\begin{equation*}
\left[ U^u\left( p',\alpha' \right)-U^u\left( p,\alpha' \right) \right]
-\left[ U^u\left( p',\alpha \right)-U^u\left( p,\alpha \right) \right]
=(\alpha'-\alpha) \left[\Delta^u \left( p\right) -\Delta^u \left(p'\right)\right] \geq 0,
\end{equation*}
because $\Delta^u \left(p\right)$ is decreasing in $p$.
\end{proof}

%---------------------------
\begin{proof}[Proof of Proposition \ref{prop: mono-condition-upper}]
For any given state $\left( p,\alpha \right) $ and a search outcome $X=x$.  The new fallback value $p'=\min \{x,p\}$ and the new belief $\alpha'$ is derived from \eqref{eq:alpha-update-upper}.  From the decomposition expression
\begin{equation*}
\begin{aligned}
U^u(p',\alpha')-U^u(p,\alpha)
&=
\underbrace{U^u(p',\alpha')-U^u(p',\alpha)}_{\text{=A: learning effect}}
+
\underbrace{U^u(p',\alpha)-U^u(p,\alpha)}_{\text{=B: fallback-price effect}} .
\end{aligned}
\end{equation*}
We are interested in deriving the conditions for \eqref{def: monotone} to hold. There are two cases.

\paragraph{Automatic cases.}
First, suppose the observation is censored, i.e., $\{X\ge p\}$. Then $p'=p$ and
\begin{equation*}
U^u\left( p',\alpha'\right) -U^u\left( p,\alpha \right)
=-\left( \alpha
^{\prime }-\alpha \right) \Delta^u(p) \leq 0,
\end{equation*}
because $\Delta^u(p)\leq 0$ and $\alpha^{\prime }-\alpha \propto \bar{F}_{1}\left( p-1\right) -\bar{F}_{2}\left(p-1\right)\leq 0$ by \eqref{eq:alpha-update-upper}.  Next, suppose the observation is fully revealed with $X=x<p$ and $x\ge \hat{x}^u$. Then $\alpha'\le \alpha$, and the learning effect $(\alpha-\alpha')\Delta^u(x)$ is negative, and hence, $U^u\left( p^{\prime
},\alpha'\right) -U^u\left( p,\alpha \right) \leq 0$.  So monotonicity holds automatically in both these cases.

\paragraph{The critical case.}
Now consider a fully revealed observation $X=x<\min\{p,\hat{x}^u\}$. Then $\alpha'>\alpha$, and the learning effect is positive and the fallback value effect is negative. The fallback value effect
\begin{equation*}
\text{B}=U^u\left( x,\alpha \right) -U^u\left( p,\alpha \right)
=-\sum_{z=x}^{p-1}F_{\alpha}(z) \le 0,
\end{equation*}
and the learning effect $
\text{A}=U^u\left( x,\alpha'\right) -U^u\left( x,\alpha \right)
=\left( \alpha -\alpha'\right) \Delta^u \left( x\right)$.
\begin{itemize}
    \item \textbf{Sufficiency}.  Suppose condition \eqref{eq: IFF_upper} holds for every $(p,\alpha)$ and every $x<\min\{p,\hat{x}^u\}$. Then, for the critical case, $A+B\le 0$. Combined with the automatic cases, we have $U^u\left( p',\alpha'\right) \le U^u\left( p,\alpha \right)$ for every feasible transition. Therefore, the search problem is monotone.  
    \item \textbf{Necessity}.  Conversely, suppose the search problem is monotone. Fix arbitrary $(y, \alpha$ and any $x<\min\{p,\hat{x}^u\}$. Monotonicity requires $U^u\left( p',\alpha'\right) \le U^u\left( p,\alpha \right)$.  Substituting the expression from the critical case, this inequality is equivalent to \eqref{eq: IFF_upper}. This proves the necessity.
\end{itemize}

We next show that \eqref{eq: SC_upper} implies condition \eqref{eq: IFF_upper}. 
Fix any $x<\min\{p,\hat{x}^u\}$. First, observe that
\begin{equation*}
\sum_{z=x}^{p-1}F_{\alpha}(z)\geq F_{\alpha}(x)\geq F_2(x),
\end{equation*}
where the second inequality follows from
$F_{\alpha}(x)=\alpha F_1(x)+(1-\alpha)F_2(x)\geq F_2(x)$ by Assumption~\ref{ass: single-crossing}.

Next, define
$r_x:=\frac{\pi_{1,x}}{\pi_{2,x}}>1$.
The posterior updating formula implies
\begin{equation*}
\alpha-\alpha'
=\frac{\alpha(1-\alpha)}{\alpha r_x+(1-\alpha)}(1-r_x)<0.
\end{equation*}
The expression on the left-hand side attains its minimum at
$\alpha^*=(1+\sqrt{r_x})^{-1}$. Since $\Delta^u(x)\leq 0$, we obtain
\begin{equation*}
(\alpha-\alpha')\Delta^u(x)
\leq
\frac{\alpha^*(1-\alpha^*)}{\alpha^*r_x+(1-\alpha^*)}
(1-r_x)\Delta^u(x)
=
\frac{1-\sqrt{r_x}}{1+\sqrt{r_x}}\Delta^u(x)
=\frac{\sqrt{\pi_{1,x}}-\sqrt{\pi_{2,x}}}{\sqrt{\pi_{1,x}}+\sqrt{\pi_{2,x}}}\,\Delta^u(x).
\end{equation*}
Combining the above inequalities yields the sufficient condition in
\eqref{eq: SC_upper}, which establishes monotonicity.
\end{proof}

%---------------------------
\begin{proof}[Proof of Proposition \ref{prop: monotone full->censoring}]
Under full revelation, suppose the monotone condition \eqref{def: monotone} holds, and for a given a prior $\a$ and a sample outcome $X$, let $\a'(X)$ denote the updated posterior given by \eqref{eq: a_t+1 full}.  Now consider lower-censoring.  For any state $(y,\a)$ and any new sample outcome $X=x$:
\begin{itemize}
    \item If $x>y$, then the value $x$ is revealed and the posterior belief $\a'$ under lower-censoring is the same as that under full revelation, and therefore, condition \eqref{def: monotone} is satisfied.
    \item If $x\le y$, then only the censored event $\{X\le y\}$ is observed. The posterior belief under lower-censoring, from \eqref{eq: a_t+1 lower} and by the law of total expectation,
    \begin{equation*}
        \a'=\P(L=L_1|\a,X\le y)=\sum_{x=1}^y\P(X=x|X\le y)\P(L=L_1|\a, X=x)=\E\left(\a'(X)|X\le y\right).
    \end{equation*}
    Since $U(y,\a)$ is linear in $\a$, we have
    \begin{equation*}
        U(y,\a')=U\left(y,\E\left(\a'(X)|X\le y\right)\right)=\E\left[U(y,\a'(X))|X\le y\right].
    \end{equation*}
    As condition \eqref{def: monotone} is satisfied under full revelation, then for all $x\le y$,  $U(y,\a)\ge U(y,\a(x))$.  Taking conditional expectation then gives $U(y,\a)\ge \E\left[U(y,\a'(X))|X\le y\right]=U(y,\a')$.
\end{itemize} 
\end{proof}

%---------------------------
\begin{proof}[Proof of Theorem \ref{thm: IFF full}]                        
We start with the sufficient part. Under condition \eqref{eq: IFF_full}, $q_x\le p_x$ for all $x\le k-1$. Given any state $(y,\a)$ and any new search outcome $X=x$, there are two cases.\\ 
Case I:  If $x\le y$, then the new fallback value $y'=y$ and
\begin{equation*}
U(y',\a')-U(y,\a)=(\a'-\a)[U_1(y)-U_2(y)]
 \propto \left( q_{x}-p_{x}\right) \Delta \left( y\right),
\end{equation*}%
because $\Delta \left( y\right)\ge 0$ for all $y\in K$ under Assumption \ref{ass: single-crossing}, i.e., $L_1\prec_{sc}L_2$, and $\a-\a'\propto q_{x}-p_{x}$ by \eqref{eq: a_t+1 full}.  If $y=k$, then $\Delta(y)=0$ and \eqref{def: monotone} is satisfied.  If $y<k$, then $x\le y\le k-1$, and hence $\a-\a' \le 0$ under condition \eqref{eq: IFF_full}. The monotone condition \eqref{def: monotone} is also satisfied.\\
   Case II:  If $x> y$, then $y'=x$ and, from the decomposition expression \eqref{eq: U(y',a')-U(y,a)},
\begin{itemize}
\item the fallback value effect B$=U\left( x,\alpha \right) -U\left( y,\alpha \right)<0$ (Lemma \ref{lem: U(y,a)}).

\item the learning effect A$=\left( \a'-\alpha \right) [U_1(x)-U_2(x)] \propto\left( q_{x}-p_{x}\right) \Delta
\left( x\right)$.  In this case, if $x\le\hat{x}$, then B is negative as $q_{x}\leq p_{x}$ under condition \eqref{eq: IFF_full}.  If $x>\hat{x}$, then $x=k$ also under condition \eqref{eq: IFF_full}, and $\Delta(x)=\Delta(k)=0$.
\end{itemize}

Next is the necessary part. Suppose the monotone condition \eqref{def: monotone} holds.  Fix any $y\in K$ and any $x\le y$.  We have $y'=y$ and $\a'$ is derived from \eqref{eq: a_t+1 full}.  It then follows from \eqref{def: monotone}  that
\begin{equation*}
U(y',\a')-U(y,\a)=(\a'-\a)[U_1(y)-U_2(y)]\propto \left( q_{x}-p_{x}\right) \Delta \left( y\right)\le 0.
\end{equation*}
For $y=k$, $\Delta(k)=0$ and the above inequality holds. For $y=k-1$, $\Delta(k-1)=q_k-p_k> 0$ and the above inequality implies, for all $x\le k-1$, $(q_{x}-p_{x})(q_{k}-p_{k})\le 0$, and then $q_{x}-p_{x}\le 0$.  This gives condition \eqref{eq: IFF_full}.
\end{proof}
%--------------------------------------

\begin{proof}[Proof of Theorem \ref{thm: sm monotone}]
We first prove the necessary and sufficient condition.  The expected monotonicity condition is $\mathbb{E}_\alpha[U(y',\alpha')-U(y,\alpha)]\le 0$, and $U(y,\alpha)=\sum_{m=y}^{k-1}\bar F_\alpha(m)$ from \eqref{eq: U(y,a) survival}.  With the decomposition expression of \eqref{eq: U(y',a')-U(y,a)}, let us consider the fallback value effect first.  If $X=x\le y$, the fallback effect is zero.  If $X=x>y$, then $y'=x$ and $U(x,\alpha)-U(y,\alpha)=-\sum_{m=y}^{x-1}\bar F_\alpha(m)$.  Thus, for the next random sample value $X$, we have
\[
U(y',\alpha)-U(y,\alpha)
=
-\ind_{{\{X>y\}}}
\sum_{m=y}^{X-1}\bar F_\alpha(m),
\]
which does not depend on whether the information structure is lower-censoring or full revelation.  The expected fallback value effect is then
\[
\mathbb{E}_\alpha[U(y',\alpha)-U(y,\alpha)]
=
-\sum_{x=y+1}^k \pi_\alpha(x)\sum_{m=y}^{x-1}\bar F_\alpha(m)
=
-\sum_{m=y}^{k-1}\bar F_\alpha(m)\sum_{x=m+1}^k \pi_\alpha(x)
=
-\sum_{m=y}^{k-1}\bar F_\alpha(m)^2.
\]
where we interchange summations in the second equality.

Using Lemma \ref{lem: U(y,a)}, the learning effect $U(y',\alpha')-U(y',\alpha)=(\alpha-\alpha')\Delta(y')$.  Under full revelation, the belief update is given by \eqref{eq: a_t+1 full}, so taking expectation gives
\begin{equation}\label{eq: expected learning effect}
  \mathbb E_\alpha[U(y',\alpha')-U(y',\alpha)]
=
\alpha(1-\alpha)
\left[
(F_2(y)-F_1(y))\Delta(y)
+
\sum_{x=y+1}^k(\pi_{2,x}-\pi_{1,x})\Delta(x)
\right].  
\end{equation}
Under lower-censoring, the belief update is given by \eqref{eq: a_t+1 lower}.  Hence, the contribution of an censored event $\{X\le y\}$ to the expected learning effect is
\[
F_\alpha(y)(\alpha-\alpha_c')\Delta(y)
=
\alpha(1-\alpha)
\bigl(F_2(y)-F_1(y)\bigr)\Delta(y).
\]
For realizations $X=x>y$, lower-censoring and full revelation coincide, and
\[
\pi_\alpha(x)(\alpha-\alpha'_x)\Delta(x)
=
\alpha(1-\alpha)
(\pi_{2,x}-\pi_{1,x})\Delta(x).
\]
Therefore, under lower-censoring, the expected learning effect is also \eqref{eq: expected learning effect}.  Consequently, in both regimes the search problem is expected monotone if and only if \eqref{eq: IFF expected}.

We next prove the state-free sufficient condition by deriving a uniform bound on both sides.  First, under lower-censoring, the learning effect can be positive only when the sample value $X=x>\hat{x}$, in which case
\[
\alpha-\alpha'_x
=
\frac{\alpha(1-\alpha)}{\pi_\alpha(x)}\bigl(\pi_{2,x}-\pi_{1,x}\bigr)
\le 
\frac{1-\sqrt{r_x}}{1+\sqrt{r_x}},
\qquad r_x=\frac{\pi_{1,x}}{\pi_{2,x}}.
\]
Therefore, an upper bound of the learning effect is
\[
\mathbb{E}_\alpha[(\alpha-\alpha')\Delta(\max\{y,X\})]
\le
\mathbb{E}_\alpha[(\alpha-\alpha')\Delta(X)\ind_{\{X> \hat{x}\}}]
\le
\sum_{x>\hat{x}} \pi_{2,x} \frac{1-\sqrt{r_x}}{1+\sqrt{r_x}} \Delta(x),
\]
where we apply $\pi_{\alpha}(x)<\pi_{2,x}$ in the last inequality.  Second, for the fallback effect, since $\bar F_\alpha(m)\ge \bar F_1(m)$ uniformly in $\alpha$.  A sufficient condition is therefore given by \eqref{eq: SC expected lower}.
\end{proof}

\newpage

\section{Online Appendix}
%---------------------------------------
\begin{theorem}\label{thm: Main-general}
The optimal stopping rule of the Bayesian search problem is determined by the smallest solution of \eqref{eq:Bellman-general}.
\end{theorem}

\begin{proof}

In the sequel, a function is referred to as a real function defined on $\bar{K} \times \Delta(\{1,\ldots,\ell\}) $. For two real functions $f,g$, we say that $f\le g$, if it occurs pointwise: for every $(y,\boldsymbol{\alpha})$, $f(y,\boldsymbol{\alpha})\le g(y,\boldsymbol{\alpha}) $. 

\begin{lemma}\label{lemma: monotonity preserving.} 
Suppose that $f,g$ are two functions such that $f\le g$. Then, $Tf\le Tg$.
\end{lemma}
   
\begin{proof}  By assumption, for every $y, X, B(\boldsymbol{\alpha}, X)$, 
$$
f\!\left(\max\{y,X\},B(\boldsymbol{\alpha}, X)\right)
\le  
g\!\left(\max\{y,X\},B(\boldsymbol{\alpha}, X)\right)
$$  
with probability $1$. Therefore, the expectation over these two random variables preserves the same inequality. Thus, $Tf\le Tg$.
\end{proof}

We now define a sequence of functions inductively on which the rest of the proof resides. Let $f_0(y,\boldsymbol{\alpha})=y$ for every $(y,\boldsymbol{\alpha})$. Then for every integer $n$, set $f_{n+1}= Tf_n$.

\begin{lemma}\label{lemma: increasing sequence.}
    The sequence  $(f_{n})_n$ is increasing (e.g., for every $\boldsymbol{\alpha}$ and $n$,  $f_{n} \le f_{n+1}$).
\end{lemma}
\begin{proof}
   We proceed by induction on $n$.  
It is immediate that $f_0 \leq f_1$, since $f_1$ is defined as the maximum of $f_0$ and something else.
Assume now that
$
f_{n-1}\leq f_{n}$.
The previous lemma implies that 
$f_{n}=
Tf_{n-1}\leq Tf_{n}=f_{n+1}$, as desired.
\end{proof}

Since $(f_n)_n$ is bounded and monotonically increasing, the limit exists. 
Denote it by $f_\infty$. Thus, $(f_n)_n\le f_\infty$ for every $n$.
\begin{lemma}\label{lemma:f_infty}
The limiting function $f_\infty$ is a fixed point of the operator $T$.
\end{lemma}

\begin{proof}
The result follows from the continuity of $T$ with respect to the supremum (maximum) topology. Formally, for every $\eps > 0$, there exists a $\delta > 0$ such that if 
$
\sup_{(y, \boldsymbol{\alpha})} |f(y, \boldsymbol{\alpha}) - g(y, \boldsymbol{\alpha})| < \delta,
$
then 
$
\sup_{(y, \boldsymbol{\alpha})} |Tf(y, \boldsymbol{\alpha}) - Tg(y, \boldsymbol{\alpha})| < \eps.
$

Let $\eps > 0$ and assume $\sup_{(y, \boldsymbol{\alpha})} |f(y, \boldsymbol{\alpha}) - g(y, \boldsymbol{\alpha})| < \eps$. This implies $f(y, \boldsymbol{\alpha}) \ge g(y, \boldsymbol{\alpha}) - \eps$ for all $(y, \boldsymbol{\alpha})$. By the monotonicity and linearity of the expectation operator, it follows that:
\begin{align*}
Tf(y,\boldsymbol{\alpha}) &= \max\left\{ y, \; \mathbb{E}_{\boldsymbol{\alpha}}\left[ f\left(\max\{y,X\}, B(\boldsymbol{\alpha},X)\right) \right] - c \right\} \\
&\ge \max\left\{ y, \; \mathbb{E}_{\boldsymbol{\alpha}}\left[ g\left(\max\{y,X\}, B(\boldsymbol{\alpha},X)\right) - \eps \right] - c \right\} \\
&= \max\left\{ y, \; \mathbb{E}_{\boldsymbol{\alpha}}\left[ g\left(\max\{y,X\}, B(\boldsymbol{\alpha},X)\right) \right] - c - \eps \right\} \\
&\ge \max\left\{ y, \; \mathbb{E}_{\boldsymbol{\alpha}}\left[ g\left(\max\{y,X\}, B(\boldsymbol{\alpha},X)\right) \right] - c \right\} - \eps \\
&= Tg(y,\boldsymbol{\alpha}) - \eps.
\end{align*}
By symmetry, swapping the roles of $f$ and $g$ yields $Tg(y,\boldsymbol{\alpha}) \ge Tf(y,\boldsymbol{\alpha}) - \eps$. Thus, $\sup_{(y, \boldsymbol{\alpha})} |Tf(y, \boldsymbol{\alpha}) - Tg(y, \boldsymbol{\alpha})| \le \eps$, establishing that $T$ is a non-expansive map (and hence continuous with $\delta = \eps$) under the supremum norm. Taking the limit as $n \to \infty$ for $f_{n+1} = Tf_n$ yields $f_\infty = Tf_\infty$.
\end{proof}

The following lemma states that $f_\infty$ is the lowest fixed point of $T$.
\begin{lemma}\label{lemma: f_infty_minimal}
If $J^r$ is a fixed point of $T$, then $f_\infty \le J^r$.
\end{lemma}
\begin{proof}
Clearly, $f_0 \le J^r$. By Lemma~\ref{lemma: monotonity preserving.} and the construction,
$f_n = T^n f_0$, we have
\[
f_n = T^n f_0 \le T^n J^r = J^r
\qquad \text{for all } n.
\]
Taking limits preserves the inequality, and hence $f_\infty \le J^r$.
\end{proof}

We now turn to finite-horizon Bayesian search problems. Suppose that the search process proceeds as usual until time $n$. If the DM has not stopped before time $n$, the existing fallback payoff is the outcome of the dynamic programming. Our goal is to relate the optimal strategy of the $n$-period search problems to the sequence $(f_n)_n$.

Fix $n$ and define the stopping time $\tau_n$ as follows. Consider any history of search outcomes of length $t$ that has not yet led to stopping, and let $(y,\boldsymbol{\alpha})$ denote the induced fallback payoff and posterior belief. The stopping rule is:
\[
\tau_n = t \quad \text{if and only if} \quad f_n(y,\boldsymbol{\alpha})=y.
\]
In other words,  $f_n(y,\boldsymbol{\alpha})=y$ implies stopping. If $f_n(y,\boldsymbol{\alpha})>y$, the stopping time remains unspecified, which corresponds to continuing the search.

The same definition applies with $f_\infty$, and the corresponding stopping time is denoted by $\tau_\infty$. Note that $\tau_n$ is formally defined on the space of outcome histories; however, the induced state $(y,\boldsymbol{\alpha})$ fully determines, via $f_n$, whether the rule prescribes stopping or continuation.

\begin{remark}\label{rm: continue in n implies in n+1}
Because the sequence $(f_n)$
is increasing, whenever $\tau_n$ dictates continuation, so does $\tau_{n+1}$ and therefore, 
 $\tau_\infty$. But it might happen that after a certain history, $\tau_n$ dictates stopping, while $\tau_{n+1}$ dictates continuation. It implies that among the  histories of a fixed length, $N$, the set of histories that follow stopping is diminishing with $n.$ 
\end{remark}

\begin{lemma}
\label{lemma:optimal tau n}
The stopping time $\tau_n$ is optimal for the $n$-period search problem.
\end{lemma}

\begin{proof}
Consider stage $n-1$ of the search process, where the state is $(y,\boldsymbol{\alpha})$. Only one search opportunity remains. The DM can either stop and receive $y$, or continue and obtain the expected payoff
$
\mathbb{E}\!\left[\max\{y,X\}\right]-c.
$
The optimal choice yields $f_1(y,\boldsymbol{\alpha})$, and $\tau_1$ implements this choice.

Now consider stage $n-2$, where two search opportunities remain, and the state is again $(y,\boldsymbol{\alpha})$. The DM can either stop with a payoff $y$, or continue and obtain
\[
\mathbb{E}\!\left[
f_1\!\left(\max\{y,X\},B(\boldsymbol{\alpha}, X)\right)
\right]-c.
\]
The maximum of these two payoffs is $f_2(y,\boldsymbol{\alpha})$, which is implemented by $\tau_2$.

Proceeding by backward induction, we conclude that $\tau_n$ is optimal for the $n$-period search problem, and that the maximal expected net payoff starting from $(y,\boldsymbol{\alpha})$ is $f_n(y,\boldsymbol{\alpha})$ (with the usual initial condition $(0,\boldsymbol{\alpha}_0)$).
\end{proof}

\begin{lemma}\label{lemma:stopping time and epsilon}
For any integer $N$ there exists an $\eps(N) >0$ such that for every $n$,
$
\P (\tau_n >N)< \eps(N),
$
and $\eps(N)\to 0$ when $N\to \infty$.
\end{lemma}
\begin{proof}
If this were not the case, then there would be $\eps>0$ such that for infinitely many $N$ there would exist $n$ such that
$
\mathbb{P}(\tau_n > N) \ge \eps .
$
In that event, with probability at least $\eps$ the search cost $c$ is incurred at least $N$ times, implying an expected payoff no greater than
$
k - \eps N c.
$
For $N$ sufficiently large, this quantity becomes negative, which contradicts 
Lemma \ref{lemma:optimal tau n}.
\end{proof}

\begin{corollary}\label{cor: proof}
For every integer $N$ there is an integer $S$ s.t.\ all $\tau_n$, $n\ge S$ and $\tau_\infty$ agree on histories shorter of $N$. 
\end{corollary}
\begin{proof} Remark~\ref{rm: continue in n implies in n+1} and Lemma~\ref{lemma:stopping time and epsilon} imply that the collection of histories at which stopping occurs is monotone nondecreasing in $n$. Since the set of histories of length less than $N$ is finite, this monotone sequence must stabilize. Consequently, there exists $S$ such that for all $n \ge S$, the stopping times coincide on every history of length less than $N$.
\end{proof}

\begin{lemma}\label{lemma:tau_infty has finite expctation}
\emph{(i)} $\E[\tau_\infty]<\infty$. \\ \emph{(ii)} $\E[\tau_\infty\mid \tau_\infty > N]$ tends to $0$ as $N\to \infty$. \end{lemma}
\begin{proof}
If $\E[\tau_\infty]=\infty$, there is an $N$ s.t.
$\E[\tau_\infty\mid \tau_\infty \le N]>3k/c$. By Corollary \ref{cor: proof}, there is an $S$ s.t.\ all $\tau_n$, $n\ge S$ and $\tau_\infty$ agree on histories shorter of $N$. It implies that the yield of 
such $\tau_n$ is not greater than 
 {
\setlength{\abovedisplayskip}{2pt}
\setlength{\belowdisplayskip}{2pt}
\begin{align*}
& \E[Y_{\tau_n}-\tau_n c|\tau_n \le N  ] \P(\tau \le N) +\P(\tau > N)k \le \\ 
&\E[Y_{\tau_n}]\P(\tau \le N) - (3k/c)c +\P(\tau > N)k\le\E[Y_{\tau_n}]\P(\tau \le N)-3k+ k<0. \nonumber
\end{align*}
}
This is impossible, and we conclude (i). (ii) is an immediate consequence of it.
\end{proof}

The next lemma shows that the minimal fixed point $f_\infty$ characterizes the optimal strategy of the infinite-horizon problem.

\begin{lemma}\label{lemma:f_infty_optimal}
The stopping time $\tau_\infty$ is optimal for the unbounded search problem.
\end{lemma}

\begin{proof}
Fix $N$.  By Lemma \ref{lemma:stopping time and epsilon} and Corollary \ref{cor: proof} there is are $N$  and $S$ s.t.\ the payoff $\tau_n$
 yields is at most
 {
\setlength{\abovedisplayskip}{1pt}
\setlength{\belowdisplayskip}{2pt}
\begin{align}\label{eq: upper bound}
&\E[Y_{\tau_n}-\tau_n c\mid \tau_n \le N] \p(\tau_n \le N) +\p(\tau_n > N)k \le \\ 
&\E[Y_{\tau_n}-\tau_n c\mid \tau _n\le N]+ \eps(N) k, \nonumber
\end{align} 
}
for  every $n\ge S$.  On the other hand,   $\tau_\infty$  yields at least
 {
\setlength{\abovedisplayskip}{3pt}
\setlength{\belowdisplayskip}{2pt}
\begin{align}\label{eq: lower bound}
&\E[X_{\tau_\infty}-\tau_\infty c\mid \tau_\infty \le N] \P(\tau_\infty \le N)+ 
\E[X_{\tau_\infty}-\tau_\infty c\mid \tau_\infty>  N] \P(\tau_\infty > N)  \ge \\
& \E[Y_{\tau_n}-\tau_\infty c\mid \tau_n \le N] \P(\tau_n\le N)- 
c\E[\tau_\infty\mid \tau_\infty>  N]. \nonumber 
\end{align}
}
The difference between \eqref{eq: upper bound} and \eqref{eq: lower bound} is at most
\[
\eps(N)\,k + c\,\E\!\left[\tau_\infty \mid \tau_\infty > N\right],
\]
which converges to $0$ as $N \to \infty$. Indeed, as $N \to \infty$: 
(a) $\eps(N) \to 0$ by Lemma~\ref{lemma:stopping time and epsilon}; and 
(b) $\E\!\left[\tau_\infty \mid \tau_\infty > N\right] \to 0$ by Lemma~\ref{lemma:tau_infty has finite expctation}.
We conclude that 
$\tau_\infty$ yields the supremum of $f_n(0, \boldsymbol{\alpha}_0)$, where $n\to \infty$, which is  $f_{\infty}(0, \boldsymbol{\alpha}_0)$

It remains to show that no stopping rule can outperform $\tau_\infty$ in the unbounded search problem. Suppose, to the contrary, that there exists a stopping time $\bar{\tau}$ and an $\varepsilon>0$ such that $\bar{\tau})$ yields at least$f_\infty(0,\boldsymbol{\alpha}_0)+\varepsilon$. Consider the truncated stopping rule $\bar{\tau}\wedge n$, which forces termination by period $n$. Since the payoff under $\bar{\tau}\wedge n$ converges to that under $\bar{\tau}$ as $n\to\infty$, there exists $N$ such that for every $n\ge N$,
$\bar{\tau}\wedge n$ yields strictly more than $ f_\infty(0,\boldsymbol{\alpha}_0)+\varepsilon/2$.
Moreover, since $f_n(0,\boldsymbol{\alpha}_0)\uparrow f_\infty(0,\boldsymbol{\alpha}_0)$, enlarging $N$ if necessary implies that 
$\bar{\tau}\wedge n$ yields at least $\ge f_n(0,\boldsymbol{\alpha}_0)+\varepsilon/2$
for all $n\ge N$. This contradicts Lemma~\ref{lemma:optimal tau n}. Therefore, no stopping rule can yield a payoff strictly exceeding $f_\infty(0,\boldsymbol{\alpha}_0)$, and we conclude that $\tau_\infty$ is an optimal stopping rule for the unbounded Bayesian search problem.
\end{proof}

The previous lemma completes the proof of Theorem~\ref{thm: Main-general}. First, the lowest fixed point of the operator $T$ is $f_\infty$, which is obtained constructively as the limit of the sequence $(f_n)_n$. Second, the stopping rule $\tau_\infty$ induced by $f_\infty$ is optimal for the unbounded search problem and constitutes the minimal solution to \eqref{eq:Bellman-general}.
\end{proof}

%------------------
\subsection{A Bayesian Search Problem with an Outside Option}
As before, the DM chooses how long to continue searching so as to maximize her expected payoff, yet she is now endowed with an outside option valued at $v_o$.  In addition to the standard choices of stopping and accepting the current fallback value or continuing the search, she may quit permanently at any time and obtain a fixed payoff $v_o$. 

At each point in time, her decision is summarized by two state variables: the maximum between the current fallback value $y$ and $v_o$ and the posterior belief $\boldsymbol{\alpha}$. 
Accordingly, given an information structure $r$ in which the updated fallback value $y'=\max\{y,v_o,X\}$ and the updated posterior is derived based on Bayes' rule and the observed search history, e.g., according to \eqref{eq: updated state}, the value function $J^r_o$ satisfies the following Bellman equation:
\begin{equation}\label{eq:Bellman-general general option}
J^r_o(\max\{y,v_o\}, \boldsymbol{\alpha})
=
\max\left\{
y,v_o,\;
\mathbb{E}_{\boldsymbol{\alpha}}\!\left[
J^r_o\!\left(
\max\{y,v_o,X\},
B(\boldsymbol{\alpha}, X)
\right)
\right]
- c
\right\}.
\end{equation}
The first two terms inside the maximization correspond to stopping immediately and obtaining the larger of the current fallback value $y$ and the outside option $v_o$. 
The third term represents the continuation value: the DM pays the search cost $c$, observes a realization $X$ drawn according to belief $\boldsymbol{\alpha}$, updates her belief to $B(\boldsymbol{\alpha}, X)$, and updates the fallback value to the best payoff available so far, namely $\max\{y,v_o,X\}$.

Finally, note that when the outside option satisfies $v_o=0$, the Bellman equation \eqref{eq:Bellman-general general option} reduces to \eqref{eq:Bellman-general general}, so the two formulations coincide.  The technique we develop for the proof of Theorem \ref{thm: Main-general} enables us to demonstrate the following result.

\begin{theorem}\label{thm: Main-general option}
The optimal stopping rule of the Bayesian search problem with outside option $v_o$ is determined by the smallest solution of \eqref{eq:Bellman-general general option}.
\end{theorem}
\begin{proof}[Proof of Theorem~\ref{thm: Main-general option}]
 As in the proof of Theorem~\ref{thm: Main-general} we define a sequence of functions inductively. Let $f_0(\max\{y,v_o\},\boldsymbol{\alpha})=\max\{y,v_o\}$ for every 
 $(y,v_o,\boldsymbol{\alpha})$. Then for every integer $n$, set $f_{n+1}= Tf_n$, where 
 \begin{equation}\label{eq:Bellman-general general operator option}
Tg(\max\{y,v_o\}, \boldsymbol{\alpha})
=
\max\left\{
y,v_o,\;
\mathbb{E}\!\left[
g\!\left(
\max (y,v_o,X\},
B(\boldsymbol{\alpha}, X)
\right)
\right]
- c
\right\}.
\end{equation}
The proof follows the footsteps of the previous one.
\end{proof}

\end{document}